\documentclass[a4paper]{llncs}
\usepackage{longtable}
\usepackage{proof}
\usepackage[table]{xcolor}
\usepackage{amsmath}
\usepackage{mathtools}
\usepackage{setspace}
\usepackage{cite}
\usepackage{xcolor}
\usepackage{float}
\usepackage{enumitem}
\usepackage{calligra}
\usepackage{graphicx}
\usepackage{algorithm,algpseudocode}
\usepackage{scalerel}
\usepackage{multirow}
\usepackage{color}
\usepackage{framed}

\makeatletter
\newcommand{\algmargin}{\the\ALG@thistlm}
\makeatother
\newlength{\whilewidth}
\algdef{SE}[parWHILE]{parWhile}{EndparWhile}[1]
{\parbox[t]{\dimexpr\linewidth-\algmargin}{%
		\hangindent\whilewidth\strut\algorithmicwhile\ #1\ \algorithmicdo\strut}}{\algorithmicend\ \algorithmicwhile}%
\algnewcommand{\parState}[1]{\State%
	\parbox[t]{\dimexpr\linewidth-\algmargin}{\strut #1\strut}}

\DeclareFontShape{T1}{calligra}{m}{n}{<->s*[2.2]callig15}{}
\DeclareMathAlphabet{\matcalligra}{T1}{calligra}{m}{n}

\DeclareMathAlphabet\mathbfcal{OMS}{cmsy}{b}{n}
\let\mathcal\undefined \DeclareMathAlphabet{\mathcal}{OMS}{cmsy}{m}{n}

\newcommand{\flqsr}{\ensuremath{\mathsf{4LQS^R}}}

\newcommand{\dlssx}{\mathcal{DL}\langle \mathsf{4LQS^{R,\!\times}}\rangle(\D)}
\newcommand{\shdlssx}{\mathcal{DL}_{\D}^{4,\!\times}}
\newcommand{\shdlss}{\mathcal{DL}_{\D}^{4}}
\newcommand{\D}{\mathbf{D}}
\newcommand{\sroiqd}{\mathcal{SROIQ}(\D)}

\newcommand{\defAs}{\coloneqq}
\newcommand{\I}{\mathbf{I}}
\newcommand{\Ind}{\mathbf{Ind}}
\newcommand{\C}{\mathbf{C}}

\newcommand{\Ra}{\mathbf{R_A}}
\newcommand{\Rd}{\mathbf{R_D}}
\newcommand{\sym}{\mathsf{Sym}}
\newcommand{\asym}{\mathsf{Asym}}
\newcommand{\refl}{\mathsf{Ref}}
\newcommand{\irref}{\mathsf{Irref}}
\newcommand{\tra}{\mathsf{Tra}}
\newcommand{\fun}{\mathsf{Fun}}

\newcommand{\vipcomment}[1]{}
\newcommand{\pow}{\mathcal{P}}
\newcommand{\var}{\mathtt{Var}}
\newcommand{\T}{\mathcal{T}}

\newcommand{\ke}{KE-tableau}
\newcommand{\M}{\mathbfcal{M}}

\newcommand{\KB}{\mathcal{KB}}
\newcommand{\vari}{\mathtt{Var}_i}
\newcommand{\varz}{\mathtt{Var}_0}
\newcommand{\varu}{\mathtt{Var}_1}

\newcommand{\vart}{\mathtt{Var}_3}
\newcommand{\vare}{\mathsf{V}_{\mathsf{e}}}

\newcommand{\vardt}{\mathsf{V}_{\mathsf{d}}}

\newcommand{\varar}{\mathsf{V}_{\mathsf{ar}}}
\newcommand{\varcon}{\mathsf{V}_{\mathsf{c}}}
\newcommand{\varind}{\mathsf{V}_{\mathsf{i}}}
\newcommand{\varcr}{\mathsf{V}_{\mathsf{cr}}}
\newcommand{\sfvar}[2]{ {\mathsf{#1}_{#2}} }

\newcommand{\DT}{\mathcal{D}}

\newcommand{\coreflqsr}{\mathsf{4LQS}_{\scriptscriptstyle{\mathcal{DL}_{ \scaleto{\mathbf{D}}{2.5pt}}^{\scaleto{4,\!\times}{3pt}}}}^R}

\newcommand{\keg}{\textnormal{KE}$^{\mathbf{\gamma}}$-tableau} 
\newcommand{\kegs}{\textnormal{KE}$^{\mathbf{\gamma}}$} 
\newcommand{\prochoplus}{\textit{HOCQA}^\gamma\textit{-}{\shdlssx}} 
\newcommand{\procho}{\textit{HOCQA}\textit{-}{\shdlssx}} 
\newcommand{\egamma}{\textnormal{E}^{\gamma}\textnormal{-rule}}
\newcommand{\seqs}{\mathcal{S}^{\overline{\beta}_i\tau}}
\newcommand{\seqsnj}{\mathcal{S}^{\overline{\beta}\tau}_j}
\newcommand{\consistency}{\textit{Consistency-}{\shdlssx}} 
\newcommand{\litqt}{Lit^{\vartheta}_q}
\newcommand{\minord}{<_{x_0}}

\title{	An optimized KE-tableau-based system for reasoning in the description logic $\shdlssx$ (Extended Version)} 

\author{Domenico Cantone \and Marianna Nicolosi-Asmundo \and \\Daniele Francesco Santamaria}
\institute{
	University of Catania, Dept. of Mathematics and Computer Science\\
	~email:~\texttt{\{cantone,nicolosi,santamaria\}@dmi.unict.it}
}

\graphicspath{{img/}}
\begin{document}
	\maketitle
	
	
\begin{abstract}
	
	We present a \ke-based procedure for the main TBox and ABox reasoning tasks for the description logic $\dlssx$, in short $\shdlssx$. The logic $\shdlssx$, representable in the decidable multi-sorted quantified set-theoretic fragment $\flqsr$, combines the high scalability and efficiency of rule languages such as the Semantic Web Rule Language (SWRL) with the expressivity of description logics. 

Our algorithm is based on a variant of the \ke\space system for sets of universally quantified clauses, where the KE-elimination rule is generalized in such a way as to incorporate the $\gamma$-rule. The novel system, called \keg, turns out to be an improvement of the system introduced in \cite{RR2017} and of standard first-order \ke x \cite{dagostino94}. Suitable benchmark test sets executed on C++ implementations of the three mentioned systems show that the performances of the \keg-based reasoner are often up to about 400\% better than the ones of the other two systems.
This a first step towards the construction of  efficient reasoners for expressive OWL ontologies based on fragments of computable set-theory.

\end{abstract}



\section{Introduction}
Recently, decidability results in Computable Set Theory have been used for knowledge representation and reasoning, in particular, in the context of description logics (DLs) and rule languages for the Semantic Web. Such efforts are motivated by the fact that there exists a natural translation function between set-theoretical fragments  and languages for the Semantic Web.

In particular, the decidable four-level stratified set-theoretic fragment $\flqsr$, involving variables of four sorts, pair terms, and a restricted form of quantification over variables of the first three sorts (cf.\ \cite{CanNic2013}) has been used in \cite{ictcs16} to represent the DL $\dlssx$, in short $\shdlssx$. 

The DL $\shdlssx$ admits Boolean operations on concepts, concept domain and range, existential quantification, and minimum cardinality on the left-hand side of inclusion axioms. It also supports role constructs such as role chains on the left hand side of inclusion axioms, Boolean operations on (abstract and concrete) roles, product of concepts, and properties on roles such as transitivity, symmetry, reflexivity, and irreflexivity. The DL $\shdlssx$  admits also data types, a simple form of concrete domains that are relevant in real world applications. In addition, it permits to express the Semantic Web Rule Language (SWRL), an extension of the Ontology Web Language (OWL). 
%
%
Decidability of the \emph{Conjunctive Query Answering} (CQA) problem for $\shdlssx$ has been proved in \cite{ictcs16} via a reduction to the CQA problem for $\flqsr$, whose decidability easily follows from that of $\flqsr$ (see \cite{CanNic2013}). In \cite{ictcs16}, the authors provided a terminating \ke\space based procedure that, given a $\shdlssx$-query $Q$ and a $\shdlssx$-knowledge base $\mathcal{KB}$ represented in set-theoretic terms, determines the answer set of $Q$ with respect to $\mathcal{KB}$. Notice that such an algorithm serves also as a decision procedure for the consistency problem for $\shdlssx$-knowledge bases (KBs). We recall that \ke\space systems \cite{dagostino1999} construct tableaux whose distinct branches define mutually exclusive situations, thus preventing the proliferation of redundant branches, typical of semantic tableaux. 

The results presented in \cite{ictcs16} have been extended in \cite{RR2017}  to the main ABox reasoning tasks for $\shdlssx$, such as instance checking and concept retrieval. by defining the Higher-Order Conjunctive Query Answering (HOCQA) problem for $\shdlssx$. 
Such problem, instantiable to the principal reasoning tasks for $\shdlssx$-ABoxes, has been defined by introducing Higher Order (HO) $\shdlssx$-conjunctive queries, admitting variables of three sorts: individual and data type variables, concept variables, and role variables. Decidability of the HOCQA problem for $\shdlssx$ has been proved  via a reduction to the HOCQA problem for the set-theoretic fragment $\flqsr$. 

In \cite{cilc17}, an implementation of the \ke\space procedure defined in \cite{RR2017} has been presented. Such prototype, written in C++, supports OWL 2  $\shdlssx$-KBs in the OWL/XML serialization. It was implemented only for TBox-reasoning services, namely, for verifying the consistency of given  ontologies. Purely universal quantifiers are eliminated by the reasoner during a preprocessing phase, in which each quantified formula is instantiated in a systematic way with the individuals of the KB. The resulting instances are then suitably handled by applying the KE-elimination and bivalence rules. In the light of the benchmarking of the prototype, it turned out that the preprocessing phase of the universally quantified formulae is more and more expensive as the size of the KB grows.

In this paper, the \ke-based procedure defined in \cite{RR2017} is modified, by eliminating the preprocessing phase for universally quantified formulae and replacing the standard KE-elimination rule with a novel elimination rule, called $\egamma$, incorporating the standard rule for treating universally quantified formulae ($\gamma$-rule). The resulting system turns out to be more efficient than the KE-system in \cite{cilc17} and the First-Order (FO) KE-system in \cite{dagostino94} as shown by suitable benchmarking tests executed on C++ implementations of the three systems. The main reason for such a speed-up relies on the fact that the novel $\egamma$ does not need to store the instances of universally quantified formulae on the \ke.



\section{Preliminaries}
\subsection{The set-theoretic fragment} \label{4LQS}

It is convenient to recall the main set-theoretic notions behind the DL $\shdlssx$ and its reasoning problems. For space reasons, we refrain from reporting the syntax and semantics of the whole $\flqsr$, as the interested reader can find it in \cite{CanNic2013} together with the decision procedure for the satisfiability problem for $\flqsr$. Thus, we restrict our attention  to the class of $\flqsr$-formulae actually involved in the set-theoretic representation of $\shdlssx$, namely propositional combinations of $\flqsr$-quantifier-free literals (atomic formulae or their negations) and $\flqsr$ purely universal formulae of the types displayed in Table \ref{tablecore}.
For the sake of  conciseness we refer to such class of $\flqsr$-formulae as $\coreflqsr$.

We recall that the fragment $\flqsr$ admits four collections, $\var_i$, of variables of sort $i$, for $i=0,1,2,3$. Variables of sort $i$, for $i=0,1,2,3$, are denoted by $X^i,Y^i,Z^i, \ldots$ (in particular, variables of sort $0$ are also denoted by $x,y,z, \ldots$). In addition to variables, also \textit{pair terms} of the form $\langle x,y \rangle$, with $x,y \in \var_0$ are allowed. Since the types of formulae illustrated in Table \ref{tablecore} do not involve variables of sort $2$, notions and definitions concerning $\coreflqsr$-formulae 
refer to variables of sorts $0,1$, and $3$ only. 

\vspace{-0.3cm}
\begin{table}[]
	\centering	\scriptsize	
	\renewcommand{\arraystretch}{1.8}
	\begin{tabular}{cc}
		\hline
		
		\multicolumn{1}{|c|}{\cellcolor{lightgray}Quantifier-free literals of level 0}      & \multicolumn{1}{c|}{\cellcolor{lightgray}Purely universal quantified formulae of level 1}                 \\ 
		\hline
		\multicolumn{1}{|c|}{$x=y, \; x \in X^1, \; \langle x,y \rangle \in X^3$} & \multicolumn{1}{c|}{\multirow{2}{*}{ \shortstack[c]{$(\forall z_1)\ldots(\forall z_n)\varphi_0$, where $z_1, \ldots, z_n \in \varz $ and $\varphi_0$ is \\any propositional combination of quantifier-free \\atomic formulae of level 0   }}} \\
		\multicolumn{1}{|c|}{$\neg (x=y), \; \neg(x \in X^1), \; \neg(\langle x,y \rangle \in X^3)$} & 	
		
		\multicolumn{1}{c|}{}                       \\
		\cline{1-2}
		&                                             \\
	\end{tabular}
	
	\caption{Types of literals and quantified formulae admitted in $\coreflqsr$.}
	\label{tablecore}
\end{table}
\vspace{-0.8cm}

The variables $z_1,\ldots,z_n$ are said to occur \textit{quantified} in $(\forall z_1) \ldots (\forall z_n) \varphi_0$. A variable occurs \textit{free} in a $\coreflqsr$-formula $\varphi$ if it does not occur quantified in any subformula of $\varphi$. For $i = 0,1,3$, we denote with $\vari(\varphi)$ the collections of variables of sort $i$ occurring free in $\varphi$.


Given sequences of distinct variables $\vec{x}$ (in $\var_0$), $\vec{X}^{1}$ (in $\var_1$), and $\vec{X}^{3}$ (in $\var_3$), of length $n$, $m$, and $q$, respectively, and sequences of (not necessarily distinct) variables $\vec{y}$ (in $\var_0$), $\vec{Y}^{1}$ (in $\var_1$), and $\vec{Y}^{3}$ (in $\var_3$), also of length $n$, $m$, and $q$, respectively, the 
$\coreflqsr$-substitution $\sigma \defAs \{ \vec{x}/\vec{y}, \vec{X}^1/\vec{Y}^1, \vec{X}^3/\vec{Y}^3\}$ is the mapping $\varphi \mapsto \varphi\sigma$ such that, for any given universal quantified $\coreflqsr$-formula $\varphi$, $\varphi\sigma$ is the $\coreflqsr$-formula obtained from $\varphi$ by replacing the free occurrences of the variable $x_i$ in $\vec{x}$ with the corresponding $y_i$ in $\vec{y}$ (for $i = 1,\ldots, n$), of $X^1_j$ in $\vec{X}^1$ with $Y^1_j$ in $\vec{Y}^1$ (for $j = 1,\ldots,m$), and of $X^3_h$ in $\vec{X}^3$ with $Y^3_h$ in $\vec{Y}^3$ (for $h= 1,\ldots,q$), respectively.  A substitution $\sigma$ is \emph{free} for $\varphi$ if the formulae $\varphi$ and $\varphi\sigma$ have exactly the same 
occurrences of quantified variables. The \emph{empty substitution}, denoted by $\epsilon$, satisfies $\varphi \epsilon = \varphi$, for each $\coreflqsr$-formula $\varphi$.

A $\coreflqsr$-\emph{interpretation} is a pair $\mathbfcal{M}=(D,M)$, where $D$ is a nonempty collection of objects (called \emph{domain} or \emph{universe} of $\mathbfcal{M}$) and $M$ is an assignment over the variables in $\mathcal{V}_i$, for $i=0,1,3$,  such that: $MX^{0} \in D,MX^1 \in \pow(D)$, and $MX^3 \in \pow(\pow(\pow(D)))$, where $ X^{i} \in \mathcal{V}_i$, for $i=0,1,3$, and $\pow(s)$ denotes the powerset of $s$.

\smallskip
\noindent
Pair terms are interpreted \emph{\`a la} Kuratowski, and therefore we put \\[.1cm]
\centerline{$M \langle x,y \rangle \defAs \{ \{ Mx \},\{ Mx,My \} \}$.}
Next, let
\begin{itemize}[topsep=0.1cm, itemsep=0cm]
	\item[-] $\mathbfcal{M}=(D,M)$ be a $\coreflqsr$-interpretation,
	
	\item[-] $x_1,\ldots,x_n \in \mathcal{V}_0$, and
	
	\item[-] $u_1, \ldots, u_n \in D$.
\end{itemize}

\smallskip
\noindent
By $\mathbfcal{M}[ \vec{x}  / \vec{u}]$, we denote the interpretation $\mathbfcal{M}'=(D,M')$ such that $M'x_i =u_i$ (for $i=1,\ldots,n$). For a $\coreflqsr$-interpretation $\mathbfcal{M} =(D,M)$ and a $\coreflqsr$-formula $\varphi$, the satisfiability relationship $ \mathbfcal{M} \models \varphi$ is recursively defined over the structure of $\varphi$ as follows. Quantifier-free literals are evaluated in a standard way according to the usual meaning of the predicates `$\in$'
and `$=$', and of the propositional negation 
`$\neg$'. Purely universal formulae are evaluated as follows:\\
\centerline{
	$\mathbfcal{M}  \models (\forall z_1) \ldots (\forall z_n) \varphi _0$ \quad iff \quad $\mathbfcal{M}   [ \vec{z} / \vec{u}] \models \varphi_0$, for all $\vec{u} \in D^{n}$.
}
Finally, compound formulae are interpreted according to the standard rules of propositional logic. If $\mathbfcal{M} \models \varphi$, then $\mathbfcal{M} $ is said to be a $\coreflqsr$-model for $\varphi$. A $\coreflqsr$-formula is said to be \emph{satisfiable} if it has a $\coreflqsr$-model. A $\coreflqsr$-formula is \emph{valid} if it is satisfied by all $\coreflqsr$-interpretations. 

\subsection{The logic $\dlssx$}\label{dlssx}
\begin{sloppypar}
In what follows we introduce the syntax and the semantics of the DL $\dlssx$ (as remarked above, more simply referred to as $\shdlssx$).
\end{sloppypar}
Let $\Ra$, $\Rd$, $\mathbf{C}$, $\mathbf{I}$ be denumerable pairwise disjoint sets of abstract role names, concrete role names, concept names, and individual names, respectively. 


Definition of data types relies on the notion of data type map, given  according to \cite{Motik2008} as follows. Let $\D = (N_{D}, N_{C},N_{F},\cdot^{\D})$ be a \emph{data type map}, where  $N_{D}$ is a finite set of data types, $N_{C}$ is a function assigning a set of constants $N_{C}(d)$ to each data type $d \in N_{D}$, $N_{F}$ is a function assigning a set of facets $N_{F}(d)$ to each $d \in N_{D}$, and $\cdot^{\D}$ is a function assigning a data type interpretation $d^{\D}$ to each  $d \in N_{D}$, a facet interpretation $f^{\D} \subseteq d^{\D}$ to each facet $f \in N_{F}(d)$, and a data value $e_{d}^{\D} \in d^{\D}$ to every constant $e_{d} \in N_{C}(d)$.  We shall assume that the interpretations of the data types in $N_{D}$ are nonempty pairwise disjoint sets.

\vipcomment{An abstract role hierarchy $\mathsf{R}_{a}^{H}$ is a finite collection of RIAs.  A strict partial order $\prec$ on  $\Ra \cup \{ R^- \mid R \in \Ra \}$ is called \emph{a regular order} if $\prec$ satisfies, additionally, $S \prec R$ iff $S^- \prec R$, for all roles R and S.\footnote{We recall that a strict partial order $\prec$  on a set $A$ is an irreflexive and transitive relation on $A$.}}

\noindent
(a) $\shdlssx$-\emph{data type}, (b) $\shdlssx$-\emph{concept}, (c) $\shdlssx$-\emph{abstract role}, and (d) $\shdlssx$-\emph{concrete role terms} are constructed according to the following syntax rules:
\begin{itemize}
	\item[(a)] $t_1, t_2 \longrightarrow dr ~|~\neg t_1 ~|~t_1 \sqcap t_2 ~|~t_1 \sqcup t_2 ~|~\{e_{d}\}\, ,$
	
	\item[(b)] $C_1, C_ 2 \longrightarrow A ~|~\top ~|~\bot ~|~\neg C_1 ~|~C_1 \sqcup C_2 ~|~C_1 \sqcap C_2 ~|~\{a\} ~|~\exists R.\mathit{Self}| \exists R.\{a\}| \exists P.\{e_{d}\}\, ,$
	
	\item[(c)] $R_1, R_2 \longrightarrow S ~|~U ~|~R_1^{-1} ~|~ \neg R_1 ~|~R_1 \sqcup R_2 ~|~R_1 \sqcap R_2 ~|~R_{C_1 |} ~|~R_{|C_1} ~|~R_{C_1 ~|~C_2} ~|~id(C) ~|~ $
	
	$C_1 \times C_2    \, ,$
	
	\item[(d)] $P_1,P_2 \longrightarrow T ~|~\neg P_1 ~|~ P_1 \sqcup P_2 ~|~ P_1 \sqcap P_2  ~|~P_{C_1 |} ~|~P_{|t_1} ~|~P_{C_1 | t_1}\, ,$
\end{itemize}
where $dr$ is a data range for $\D$, $t_1,t_2$ are data type terms, $e_{d}$ is a constant in $N_{C}(d)$, $a$ is an individual name, $A$ is a concept name, $C_1, C_2$ are $\shdlssx$-concept terms, $S$ is an abstract role name, $U$ is an abstract role name denoting the universal role,  $R, R_1,R_2$ are $\shdlssx$-abstract role terms, $T$ is a concrete role name, and $P,P_1,P_2$ are $\shdlssx$-concrete role terms. We remark that data type terms are introduced in order to represent derived data types.

A $\shdlssx$-KB is a triple ${\mathcal K} = (\mathcal{R}, \mathcal{T}, \mathcal{A})$ such that $\mathcal{R}$ is a $\shdlssx$-$RBox$, $\mathcal{T}$ is a $\shdlssx$-$TBox$, and $\mathcal{A}$ a $\shdlssx$-$ABox$. 

A $\shdlssx$-$RBox$ is a collection of statements of the following forms:
\[
\begin{array}{cccccccccc}
R_1 \equiv R_2,~&R_1 \sqsubseteq R_2,~&~~R_1\ldots R_n \sqsubseteq R_{n+1}~~,~&\sym(R_1),~&\asym(R_1),\\
\refl(R_1),~& \irref(R_1),~&\mathsf{Dis}(R_1,R_2),~&\tra(R_1),~&\fun(R_1),\\
R_1 \equiv C_1 \times C_2,~&P_1 \equiv P_2,~&P_1 \sqsubseteq P_2,~&\mathsf{Dis}(P_1,P_2),~&\fun(P_1),
\end{array}
\]
where $R_1,R_2$ are $\shdlssx$-abstract role terms, $C_1, C_2$ are $\shdlssx$-abstract concept terms, and $P_1,P_2$ are $\shdlssx$-concrete role terms. Any expression of the type $w \sqsubseteq R$, where $w$ is a finite string of $\shdlssx$-abstract role terms and $R$ is an $\shdlssx$-abstract role term, is called a \emph{role inclusion axiom (RIA)}. 

A $\shdlssx$-$TBox$ is a set of statements of the types:
\begin{itemize}
	\item[-] $C_1 \equiv C_2$, $C_1 \sqsubseteq C_2$, $C_1 \sqsubseteq \forall R_1.C_2$, $\exists R_1.C_1 \sqsubseteq C_2$, $\geq_n\!\! R_1. C_1 \sqsubseteq C_2$, \\$C_1 \sqsubseteq {\leq_n\!\! R_1. C_2}$,
	\item[-] $t_1 \equiv t_2$, $t_1 \sqsubseteq t_2$, $C_1 \sqsubseteq \forall P_1.t_1$, $\exists P_1.t_1 \sqsubseteq C_1$, $\geq_n\!\! P_1. t_1 \sqsubseteq C_1$, $C_1 \sqsubseteq {\leq_n\!\! P_1. t_1}$,
\end{itemize}
where $C_1,C_2$ are $\shdlssx$-concept terms, $t_1,t_2$ data type terms, $R_1$  a $\shdlssx$-abstract role term, $P_1$ a $\shdlssx$-concrete role term. Any statement of the form $C \sqsubseteq D$, with  $C$, $D$ $\shdlss$-concept terms, is a 
\emph{general concept inclusion axiom}.

A $\shdlssx$-$ABox$ is a set of \emph{individual assertions} of the forms: $a : C_1$, $(a,b) : R_1$, 
$a=b$, $a \neq b$, $e_{d} : t_1$, $(a, e_{d}) : P_1$, 
with $C_1$ a $\shdlssx$-concept term, $d$ a data type, $t_1$ a data type term, $R_1$ a $\shdlssx$-abstract role term, $P_1$ a $\shdlssx$-concrete role term, $a,b$ individual names, and $e_{d}$ a constant in $N_{C}(d)$.

The semantics of $\shdlssx$ is given by means of an interpretation $\I= (\Delta^\I, \Delta_{\D}, \cdot^\I)$, where $\Delta^\I$ and $\Delta_{\D}$ are non-empty disjoint domains such that $d^\D\subseteq \Delta_{\D}$, for every $d \in N_{D}$, and
$\cdot^\I$ is an interpretation function. The definition of the interpretation of concepts and roles, axioms, and assertions is shown in Table \ref{semdlss}.

{\small
	\begin{longtable}{|>{\centering}m{2.5cm}|c|>{\centering\arraybackslash}m{6.7cm}|}
		\hline
		Name & Syntax & Semantics \\
		\hline
		
		concept & $A$ & $ A^\I \subseteq \Delta^\I$\\
		
		ab. (resp., cn.) rl. & $R$ (resp., $P$ )& $R^\I \subseteq \Delta^\I \times \Delta^\I$ \hspace*{0.5cm} (resp., $P^\I \subseteq \Delta^\I \times \Delta_\D$)\\
		
		
		individual& $a$& $a^\I \in \Delta^\I$\\
		
		nominal & $\{a\}$ & $\{a\}^\I = \{a^\I \}$\\
		
		dtype  (resp., ng.) & $d$ (resp., $\neg d$)& $ d^\D \subseteq \Delta_\D$ (resp., $\Delta_\D \setminus d^\D $)\\
		
		
		negative data type term & $ \neg t_1 $ & $  (\neg t_1)^{\D} = \Delta_{\D} \setminus t_1^{\D}$ \\
		
		data type terms intersection & $ t_1 \sqcap t_2 $ & $  (t_1 \sqcap t_2)^{\D} = t_1^{\D} \cap t_2^{\D} $ \\
		
		data type terms union & $ t_1 \sqcup t_2 $ & $  (t_1 \sqcup t_2)^{\D} = t_1^{\D} \cup t_2^{\D} $ \\
		
		constant in $N_{C}(d)$ & $ e_{d} $ & $ e_{d}^\D \in d^\D$ \\
		
		
		
		\hline
		data range  & $\{ e_{d_1}, \ldots , e_{d_n} \}$& $\{ e_{d_1}, \ldots , e_{d_n} \}^\D = \{e_{d_1}^\D \} \cup \ldots \cup \{e_{d_n}^\D \} $ \\
		
		data range   &  $\psi_d$ & $\psi_d^\D$\\
		
		data range    & $\neg dr$ &  $\Delta_\D \setminus dr^\D $\\
		
		\hline
		
		top (resp., bot.) & $\top$ (resp., $\bot$ )& $\Delta^\I$  (resp., $\emptyset$)\\
		
		
		negation & $\neg C$ & $(\neg C)^\I = \Delta^\I \setminus C$ \\
		
		conj. (resp., disj.) & $C \sqcap D$ (resp., $C \sqcup D$)& $ (C \sqcap D)^\I = C^\I \cap D^\I$  (resp., $ (C \sqcup D)^\I = C^\I \cup D^\I$)\\
		
		
		valued exist. quantification & $\exists R.{a}$ & $(\exists R.{a})^\I = \{ x \in \Delta^\I : \langle x,a^\I \rangle \in R^\I  \}$ \\
		
		data typed exist. quantif. & $\exists P.{e_{d}}$ & $(\exists P.e_{d})^\I = \{ x \in \Delta^\I : \langle x, e^\D_{d} \rangle \in P^\I  \}$ \\

		self concept & $\exists R.\mathit{Self}$ & $(\exists R.\mathit{Self})^\I = \{ x \in \Delta^\I : \langle x,x \rangle \in R^\I  \}$ \\
		
		nominals & $\{ a_1, \ldots , a_n \}$& $\{ a_1, \ldots , a_n \}^\I = \{a_1^\I \} \cup \ldots \cup \{a_n^\I \} $ \\
		
		\hline
		
		universal role & U & $(U)^\I = \Delta^\I \times \Delta^\I$\\
		
		inverse role & $R^-$ & $(R^-)^\I = \{\langle y,x \rangle  \mid \langle x,y \rangle \in R^\I\}$\\
		
		concept cart. prod. & $ C_1 \times C_2$   &  $ (C_1 \times C_2)^I = C_1^I \times C_2^I$ \\
		
		abstract role complement & $ \neg R $ & $ (\neg R)^\I=(\Delta^\I \times \Delta^\I) \setminus R^\I $\\
		
		abstract role union & $R_1 \sqcup R_2$ & $ (R_1 \sqcup R_2)^\I = R_1^\I \cup R_2^\I $\\
		
		abstract role intersection & $R_1 \sqcap R_2$ & $ (R_1 \sqcap R_2)^\I = R_1^\I \cap R_2^\I $\\
		
		abstract role domain restr. & $R_{C \mid }$ & $ (R_{C \mid })^\I = \{ \langle x,y \rangle \in R^\I : x \in C^\I  \} $\\

		concrete role complement & $ \neg P $ & $ (\neg P)^\I=(\Delta^\I \times \Delta^\D) \setminus P^\I $\\
		
		concrete role union & $P_1 \sqcup P_2$ & $ (P_1 \sqcup P_2)^\I = P_1^\I \cup P_2^\I $\\
		
		concrete role intersection & $P_1 \sqcap P_2$ & $ (P_1 \sqcap P_2)^\I = P_1^\I \cap P_2^\I $\\
		
		concrete role domain restr. & $P_{C \mid }$ & $ (P_{C \mid })^\I = \{ \langle x,y \rangle \in P^\I : x \in C^\I  \} $\\
		
		concrete role range restr. & $P_{ \mid t}$ &  $ (P_{\mid t})^\I = \{ \langle x,y \rangle \in P^\I : y \in t^\D  \} $\\
		
		concrete role restriction & $P_{ C_1 \mid t}$ &  $ (P_{C_1 \mid t})^\I = \{ \langle x,y \rangle \in P^\I : x \in C_1^\I \wedge y \in t^\D  \} $\\
		
		\hline
		
		concept subsum. & $C_1 \sqsubseteq C_2$ & $\I \models_\D C_1 \sqsubseteq C_2 \; \Longleftrightarrow \; C_1^\I \subseteq C_2^\I$ \\
		
		ab. role subsum. & $ R_1 \sqsubseteq R_2$ & $\I \models_\D R_1 \sqsubseteq R_2 \; \Longleftrightarrow \; R_1^\I \subseteq R_2^\I$\\
		
		role incl. axiom & $R_1 \ldots R_n \sqsubseteq R$ & $\I \models_\D R_1 \ldots R_n \sqsubseteq R  \; \Longleftrightarrow \; R_1^\I\circ \ldots \circ R_n^\I \subseteq R^\I$\\
		cn. role subsum. & $ P_1 \sqsubseteq P_2$ & $\I \models_\D P_1 \sqsubseteq P_2 \; \Longleftrightarrow \; P_1^\I \subseteq P_2^\I$\\
		
		\hline
		
		symmetric role & $\sym(R)$ & $\I \models_\D \sym(R) \; \Longleftrightarrow \; (R^-)^\I \subseteq R^\I$\\
		
		asymmetric role & $\asym(R)$ & $\I \models_\D \asym(R) \; \Longleftrightarrow \; R^\I \cap (R^-)^\I = \emptyset $\\
		
		transitive role & $\tra(R)$ & $\I \models_\D \tra(R) \; \Longleftrightarrow \; R^\I \circ R^\I \subseteq R^\I$\\
		
		disj. ab. role & $\mathsf{Dis}(R_1,R_2)$ & $\I \models_\D \mathsf{Dis}(R_1,R_2) \; \Longleftrightarrow \; R_1^\I \cap R_2^\I = \emptyset$\\
		
		reflexive role & $\refl(R)$& $\I \models_\D \refl(R) \; \Longleftrightarrow \; \{ \langle x,x \rangle \mid x \in \Delta^\I\} \subseteq R^\I$\\
		
		irreflexive role & $\irref(R)$& $\I \models_\D \irref(R) \; \Longleftrightarrow \; R^\I \cap \{ \langle x,x \rangle \mid x \in \Delta^\I\} = \emptyset  $\\
		
		func. ab. role & $\fun(R)$ & \scriptsize{$\I \models_\D \fun(R) \; \Longleftrightarrow \; (R^{-})^\I \circ R^\I \subseteq  \{ \langle x,x \rangle \mid x \in \Delta^\I\}$}  \\
		
		disj. cn. role & $\mathsf{Dis}(P_1,P_2)$ & $\I \models_\D \mathsf{Dis}(P_1,P_2) \; \Longleftrightarrow \; P_1^\I \cap P_2^\I = \emptyset$\\
		
		func. cn. role & $\fun(P)$ & $\I \models_\D \fun(p) \; \Longleftrightarrow \; \langle x,y \rangle \in P^\I \mbox{ and } \langle x,z \rangle \in P^\I \mbox{ imply } y = z$  \\
		
		\hline
		
		data type terms equivalence & $ t_1 \equiv t_2 $ & $ \I \models_{\D} t_1 \equiv t_2 \Longleftrightarrow t_1^{\D} = t_2^{\D}$\\
		
		data type terms diseq. & $ t_1 \not\equiv t_2 $ & $ \I \models_{\D} t_1 \not\equiv t_2 \Longleftrightarrow t_1^{\D} \neq t_2^{\D}$\\
		
		data type terms subsum. & $ t_1 \sqsubseteq t_2 $ &  $ \I \models_{\D} (t_1 \sqsubseteq t_2) \Longleftrightarrow t_1^{\D} \subseteq t_2^{\D} $ \\
		
		\hline
		
		concept assertion & $a : C_1$ & $\I \models_\D a : C_1 \; \Longleftrightarrow \; (a^\I \in C_1^\I) $ \\
		
		agreement & $a=b$ & $\I \models_\D a=b \; \Longleftrightarrow \; a^\I=b^\I$\\
		
		disagreement & $a \neq b$ & $\I \models_\D a \neq b  \; \Longleftrightarrow \; \neg (a^\I = b^\I)$\\
		
		
		ab. role asser. & $ (a,b) : R $ & $\I \models_\D (a,b) : R \; \Longleftrightarrow \;  \langle a^\I , b^\I \rangle \in R^\I$ \\
		
		cn. role asser. & $ (a,e_d) : P $ & $\I \models_\D (a,e_d) : P \; \Longleftrightarrow \;   \langle a^\I , e_d^\D \rangle \in P^\I$ \\

		\hline \caption{Semantics of $\shdlssx$.}\\
		\caption*{\emph{Legenda.} \emph{ab.}: abstract, \emph{cn.}: concrete, \emph{rl.}: role, \emph{ind.}: individual, \emph{d. cs.}: data type constant, \emph{dtype}: data type, \emph{ng.}: negated, \emph{bot.}: bottom, \emph{incl.}: inclusion, \emph{asser.}: assertion.}  \label{semdlss}
\end{longtable}}

Let $\mathcal{R}$, $\mathcal{T}$, and $\mathcal{A}$  be as above. An interpretation $\I= (\Delta ^ \I, \Delta_{\D}, \cdot ^ \I)$ is a $\D$-model of $\mathcal{R}$ (resp., $\mathcal{T}$), and we write $\I \models_{\D} \mathcal{R}$ (resp., $\I \models_{\D} \mathcal{T}$), if $\I$ satisfies each axiom in $\mathcal{R}$ (resp., $\mathcal{T}$) according to the semantic rules in Table \ref{semdlss}.  Analogously,  $\I= (\Delta^ \I, \Delta_{\D}, \cdot^\I)$ is a $\D$-model of $\mathcal{A}$, and we write $\I \models_{\D} \mathcal{A}$, if $\I$ satisfies each assertion in $\mathcal{A}$, according to the semantic rules in Table \ref{semdlss}. 

A $\shdlssx$-KB $\mathcal{K}=(\mathcal{A}, \mathcal{T}, \mathcal{R})$ is consistent if there exists an interpretation $\I= (\Delta^ \I, \Delta_{\D}, \cdot^\I)$ that is a $\D$-model of $\mathcal{A}$,  $\mathcal{T}$, and $\mathcal{R}$.

\subsubsection{Expressiveness of the DL $\shdlssx$.}

Despite the fact that the description logic $\shdlssx$ is limited as far as the introduction of new individuals is concerned, 
%
it is more liberal than $\sroiqd\space$ \cite{Horrocks2006} in the construction of role inclusion axioms, since the roles involved are not restricted by any ordering relationship, the notion of simple role is not needed, and Boolean operations on roles and role constructs such as the product of concepts  are admitted.
Moreover, 
$\shdlssx$ supports more OWL constructs than the DLs underpinning the profiles OWL QL, OWL RL, and OWL EL \cite{Krotzsch2012}, such as disjoint union of concepts and union of data ranges. Furthermore, basic and derived data types can be used inside inclusion axioms involving concrete roles. In addition, concerning the expressiveness of rules, the set-theoretic fragment $\coreflqsr$ underpinning $\shdlssx$ allows one to express the disjunctive Datalog fragment  admitting negation, equality and constraints, subject to no safety condition, and supporting for data types.

\subsubsection{Reasoning with the DL $\shdlssx$.}  
Next, we introduce the reasoning services available for the DL $\shdlssx$, i.e., the type of inferences that can be drawn from what is explicitly asserted in a $\shdlssx$-KB. 
Specifically, we focus on two families of reasoning tasks, one concerning TBoxes and the other one regarding ABoxes. Among the main TBox reasoning problems, such as \emph{satisfiability of a concept}, \emph{subsumption of concepts}, \emph{equivalence of concepts}, and \emph{disjunction of concepts}, the problem of deciding the consistency of a $\shdlssx$-KB is the most representative one, since it comprehends the majority of them.\footnote{A separate analysis is required by the classification problem of a TBox, consisting in the computation of ancestor and descendant concepts of a given concept in a TBox, and that will be addressed in a future work.} In \cite{ictcs16} we proved the decidability of   the consistency problem of a $\shdlssx$-KB and of a relevant ABox reasoning task, namely the \emph{Conjunctive Query Answering} (CQA) problem for $\shdlssx$ consisting in computing the answer set of a $\shdlssx$-conjunctive query with respect to a $\shdlssx$-KB. In \cite{RR2017} we generalized the problem introducing the \emph{Higher Order Conjuctive Query Answering} (HOCQA) problem for $\shdlssx$. Such problem is characterized by \emph{Higher Order} (HO) $\shdlssx$-conjunctive queries  admitting variables of three sorts: individual and data type variables, concept variables, and role variables. It consists in finding the HO-answer set of a HO $\shdlssx$-conjunctive query with respect to a $\shdlssx$-KB.


Specifically, let $\varind  = \{\sfvar{v}{1}, \sfvar{v}{2}, \ldots\}$, 
$\vare = \{\sfvar{e_1}, \sfvar{e_2}, \ldots\}$,
$\vardt =\{\sfvar{t_1}, \sfvar{t_2}, \ldots \}$, 
$\varcon = \{\sfvar{c}{1}, \sfvar{c}{2}, \ldots\}$, 
$\varar = \{\sfvar{r}{1}, \sfvar{r}{2}, \ldots\}$,  
and $\varcr  = \{\sfvar{p}{1},$ $\sfvar{p}{2}, \ldots\}$ be pairwise disjoint denumerably infinite sets of variables 
disjoint from $\Ind$, $\bigcup\{N_C(d): d \in N_{\D}\}$, $\C$, $\Ra$, and $\Rd$. HO-$\shdlssx$-\emph{atomic formulae} are expressions of the following types:
\[
R(w_1,w_2),
P(w_1, u),
C(w_1),
t(u),
\mathsf{r}(w_1,w_2),
\mathsf{p}(w_1, u),
\mathsf{c}(w_1),
\sfvar{t}{} (u),
w_1=w_2,
\] 
where $w_1,w_2 \in \varind \cup \Ind$, $u \in  \vare \cup \bigcup \{N_C(d): d \in N_{\D}\}$, $R$ is a $\shdlssx$-abstract role term, $P$ is a $\shdlssx$-concrete role term, $C$ is a $\shdlssx$-concept term,  $t$ is a $\shdlssx$-data type term, $\mathsf{r} \in \varar$, $\mathsf{p} \in \varcr$, $\mathsf{c} \in \varcon$, $\sfvar{t}{} \in \vardt$. A HO $\shdlssx$-atomic formula containing no variables is said to be \emph{ground}. A HO $\shdlssx$-\emph{literal} is a HO $\shdlssx$-atomic formula or its negation. 
A HO $\shdlssx$-\emph{conjunctive query} is a conjunction of HO $\shdlssx$-literals. 
We denote with $\lambda$ the \emph{empty} HO $\shdlssx$-conjunctive query.

Let  $\sfvar{v}{1},\ldots, \sfvar{v}{n} \in \varind$, $\sfvar{e}{1},\ldots,\sfvar{e}{g} \in \vare$, $\sfvar{t}{1},\ldots,\sfvar{t}{l} \in \vardt$,$\sfvar{c}{1}, \ldots, \sfvar{c}{m} \in \varcon$, $\sfvar{r}{1}, \ldots, \sfvar{r}{k} \in \varar$, $\sfvar{p}{1}, \ldots, \sfvar{p}{h} \in \varcr$, $o_1, \ldots, o_n \in \Ind$, $e_{d_1}, \ldots, e_{d_g} \in  \bigcup \{N_C(d): d \in N_{\D}\}$, $C_1, \ldots, C_m \in \C$, $R_1, \ldots, R_k \in \Ra$, and $P_1, \ldots, P_h \in \Rd$.
A substitution 	$\sigma  \defAs \{\sfvar{v}{1}/o_1, \ldots, \sfvar{v}{n}/o_n, \sfvar{e}{1}/{e_{d_1}}, \ldots, \sfvar{e}{g}/{e_{d_g}},  \sfvar{t}{1}/t_1, \ldots, \sfvar{t}{l}/t_l,\sfvar{c}{1}/{C_1}, \ldots, \sfvar{c}{m}/{C_m},\\ \null ~~\sfvar{r}{1}/{R_1}, \ldots, \sfvar{r}{k}/{R_k}, \sfvar{p}{1} /{P_1}, \ldots, \sfvar{p}{h}/{P_h} \}$ 	
is a map such that, for every  HO-$\shdlssx$-literal $L$, $L\sigma$ is obtained from $L$ by replacing:
the occurrences of $\sfvar{v}{i}$ in $L$ with $o_i$, for $i=1, \ldots, n$,
the occurrences of $\sfvar{e}{b}$ in $L$ with $d_b$, for $b=1, \ldots, g$,
the occurrences of $\sfvar{t}{s}$ in $L$ with $t_s$, for $s=1, \ldots, l$,
the occurrences of $\sfvar{c}{j}$ in $L$ with $C_j$, for $j=1, \ldots, m$,
the occurrences of $\sfvar{r}{\ell}$ in $L$ with $R_\ell$, for $\ell=1, \ldots, k$,
the occurrences of $\sfvar{p}{t}$ in $L$ with $P_t$, for $t=1, \ldots, h$.

Substitutions can be extended to HO $\shdlssx$-conjunctive queries in the usual way. 
Let $Q \defAs  (L_1 \wedge \ldots \wedge L_m)$ be a HO $\shdlssx$-conjunctive query, and $\KB$ a $\shdlssx$-KB. A substitution $\sigma$ involving \emph{exactly} the variables occurring in $Q$ is a \emph{solution for $Q$ w.r.t.\ $\KB$} if there exists a $\shdlssx$-interpretation $\I$ such that $\I \models_{\D} \KB$ and $\I \models_{\D} Q \sigma$. The collection $\Sigma$ of the  solutions for $Q$ w.r.t.\ $\KB$ is the \emph{HO-answer set of $Q$ w.r.t.\ $\KB$}. Then the \emph{HOCQA problem} for $Q$ w.r.t.\ $\KB$ consists in finding the HO-answer set $\Sigma$ of $Q$ w.r.t.\ $\KB$. 

As illustrated in \cite{RR2017}, the HOCQA problem can be instantiated to significant ABox reasoning problems such as 
 (A) \emph{role filler retrieval}, the problem of retrieving all the fillers $x$ such that the pair $(a,x)$ is an instance of a role $R$; (B) \emph{concept retrieval}, the problem of retrieving all concepts which an individual is an instance of; (C) \emph{role instance retrieval}, the problem of retrieving all roles which a pair of individuals $(a,b)$ is an instance of; and (D) \emph{conjunctive query answering}, the problem of finding the answer set of a conjunctive query. 


In \cite{RR2017} we solved the HOCQA problem just stated by reducing it to the analogous problem formulated in the context of the fragment $\coreflqsr$ (and in turn to the decision procedure for $\flqsr$ presented in \cite{CanNic2013}).

The HOCQA problem for $\coreflqsr$-formulae can be stated as follows.
Let $\phi$ be a $\coreflqsr$-formula and $\psi$ a conjunction of $\coreflqsr$-quantifier-free literals. 
The \emph{HOCQA problem for $\psi$ w.r.t.\\ $\phi$} consists in computing the HO \emph{answer set of $\psi$ w.r.t.\\ $\phi$}, namely the collection $\Sigma'$ of all the  substitutions $\sigma'$ such that  $\M \models \phi \wedge \psi\sigma'$, for some $\coreflqsr$-interpretation $\M$.


In view of the decidability of the satisfiability problem for $\flqsr$-formulae, the HOCQA problem for $\coreflqsr$-formulae is decidable as well. 


The reduction is carried out by means of a 
function $\theta$  that maps the $\shdlssx$-KB $\KB$ in a $\coreflqsr$-formula $\phi_{\KB}$ in Conjunctive Normal Form (CNF) and the HO $\shdlssx$-conjunctive query $Q$ in the $\coreflqsr$-formula $\psi_{Q}$.  Specifically,\footnote{The map $\theta$ coincides with the transformation function defined in \cite{ictcs16} as far as it concerns the translation of each axiom or assertion $H$ of  $\KB$ in a set-theoretic formula $\theta(H)$. The map $\theta$  extends the function introduced in \cite{ictcs16} as far as it concerns the translation of the HO query $Q$ and of the substitutions $\sigma$ of the HO-answer set $\Sigma$. In particular, it maps effectively variables in $\varcon$ in variables of sort 1 (in the language of $\coreflqsr$), and variables in $\varar$ and in $\varcr$ in variables of sort 3. $\xi_1$--$\xi_{12}$ are constraints added to make sure that each $\coreflqsr$-model of $\phi_{\KB}$ can be transformed into a $\shdlssx$-interpretation (cf.\ \cite[Theorem~ 1]{ictcs16}). }

\centerline{
	$\phi_{\KB} \defAs \bigwedge_{H \in \KB} \theta(H) \wedge \bigwedge_{i=1}^{12} \xi_i, \qquad \psi_Q \defAs \theta(Q)\,.$
}
Let $\Sigma$ be the HO-answer set of $Q$ w.r.t.\\ $\KB$ and  $\Sigma'$ the HO-answer set of $\psi_Q$ w.r.t.\\ $\phi_{\KB}$. Then $\Sigma$ consists of all substitutions $\sigma$ (involving exactly the variables occurring in $Q$) such that $\theta(\sigma) \in \Sigma'$.
By Lemma 1 in \cite{RR2017}, $\Sigma'$ can be calculated effectively and thus $\Sigma$ can be calculated effectively as well.


\section{A  \ke\space based algorithm for reasoning in $\shdlssx$ }

In what follows, we introduce various definitions and notations useful for the presentation of  the procedures $\consistency$ and $\prochoplus$. 
The procedure $\consistency$ takes as input a $\coreflqsr$-formula  $\phi_{\KB}$ representing a $\shdlssx$-KB  and checks its consistency. If $\phi_{\KB}$ is consistent, the procedure $\consistency$ builds a \ke\space $\T_{\KB}$ whose distinct open and complete branches induce the models of $\phi_{\KB}$.  
 Then the procedure $\prochoplus$ computes the answer set of a given $\coreflqsr$-formula $\psi_Q$, representing a $\shdlssx$-HO conjunctive query $Q$, with respect to $\phi_{\KB}$ by means of a forest of decision trees based on the branches of the \ke\space $\T_\KB$ computed by the procedure $\consistency$ with input $\phi_\KB$. 

We recall that \ke\space is a refutation system inspired to Smullyan's semantic tableaux \cite{smullyan1995first} (see \cite{dagostino1999} for details). It differs from the latter because it includes an analytic cut rule (PB-rule)  minimizing the inefficiencies of semantic tableaux. In fact, firstly, the classic tableau system cannot represent the use of auxiliary lemmas in proofs; secondly, it cannot express the bivalence of classical logic. Thirdly, it is extremely inefficient, as shown by the fact that it cannot polynomially simulate the truth-tables. If the cut rule is admitted, none of these anomalies occurs. Before defining the procedures to be given next, we shortly introduce a variant of \ke\space called \keg.  

%
%
%

Let $\Phi \defAs \{ C_1,\ldots, C_p\}$, where each $C_i$ is either a $\coreflqsr$-quantifier free literal of the types illustrated in Table \ref{tablecore} or a $\coreflqsr$-purely universal quantified formula of the form $(\forall{x_1})\ldots(\forall{x_m})(\beta_1 \vee \ldots \vee \beta_n)$, where $\beta_1,\ldots, \beta_n$ are $\coreflqsr$-quantifier free literals. $\mathcal{T}$ is a \emph{\keg} for $\Phi$ if there exists a finite sequence $\mathcal{T}_1, \ldots, \mathcal{T}_t$ such that (i) $\mathcal{T}_1$ is the one-branch tree consisting of the sequence $C_1,\ldots, C_p$, (ii) $\mathcal{T}_t = \mathcal{T}$, and (iii) for each $i<t$, $\mathcal{T}_{i+1}$ is obtained from $\mathcal{T}_i$ either by an application of one of the rules ($\egamma$ or PB-rule) in Fig.~\ref{exprule} or by applying a substitution $\sigma$ to a branch $\vartheta$ of $\mathcal{T}_i$ (in particular, the substitution $\sigma$ is applied to each formula $X$ of $\vartheta$ and the resulting branch will be denoted with $\vartheta\sigma$). In the definition of the $\egamma$ reported in Fig.~\ref{exprule}: (a) $\tau :=\{x_1/x_{o_1} \ldots x_m/x_{o_m}\}$ is a substitution such that $x_1,\ldots,x_m$ are the quantified variables in $\psi$ and $x_{o_1}, \ldots, x_{o_m} \in \varz(\Phi)$; (b) $\seqs \defAs \{ \overline{\beta}_1\tau,\ldots,\overline{\beta}_n\tau\} \setminus \{\overline{\beta}_i\tau\}$ is a set containing the  complements of all the disjuncts $\beta_1, \ldots, \beta_n$ to which the substitution $\tau$ is applied, with the exception of the disjunct $\beta_i$.  
\vspace*{-0.5cm}
\begin{figure}
	{{\footnotesize
			\begin{center}
				\begin{minipage}[h]{5cm}
					$\infer[\textbf{E}^{\mathbf{\gamma}}\textbf{-rule}]
					{\beta_i\tau}{\psi & \quad \seqs}$\\[.1cm]
					{ where\\[-.1cm] $\psi=(\forall{x_1})\ldots(\forall{x_m})(\beta_1 \vee \ldots \vee \beta_n)$,\\
						$\tau :=\{x_1/x_{o_1} \ldots x_m/x_{o_m}\}$,\\
						and $\seqs \defAs \{ \overline{\beta}_1\tau,...,\overline{\beta}_n\tau\} \setminus \{\overline{\beta}_i\tau\}$,}{ for $i=1,...,n$}
				\end{minipage}~~~~~~~~~~~
				\begin{minipage}[h]{3.5cm}
					\vspace{-1.12cm}
					$\infer[\textbf{PB-rule}]
					{A~~|~~\overline{A}}{}$\\[.1cm]
					{ where $A$ is a literal}
				\end{minipage}
			\end{center}
			\vspace{-.2cm}
		}
		\caption{\label{exprule} Expansion rules for the \keg.}
	}
\end{figure}
\vspace*{-0.5cm}
%
%

Let $\mathcal{T}$ be a \keg. A branch $\vartheta$ of $\mathcal{T}$  is \textit{closed} if either it contains both $A$ and $\neg A$, for some formula $A$, or a literal of type $\neg(x = x)$. Otherwise, the branch is \textit{open}. A \keg\space is \emph{closed} if all its branches are closed. 
A formula $\psi=(\forall{x_1})\ldots(\forall{x_m})(\beta_1 \vee \ldots \vee \beta_n)$ is \textit{fulfilled} in a branch $\vartheta$, if  $\vartheta$ contains $\beta_i\tau$ for some $i=1,\ldots,n$ and for all $\tau$ having as domain the set $Q\var_0(\psi)=\{x_1, \ldots, x_m\}$ of the quantified variables occurring in $\psi$, and as range the set $\var_0(\vartheta)$ of the variables of sort 0 occurring free in $\vartheta$. Notice that since the procedure $\consistency$ to be defined next does not introduce any new variable, $\var_0(\vartheta)$ coincides with $\var_0(\phi_{\KB})$, for every branch $\vartheta$. A branch $\vartheta$ is \textit{fulfilled} if every formula $\psi=(\forall{x_1})\ldots(\forall{x_m})(\beta_1 \vee \ldots \vee \beta_n)$ occurring in $\vartheta$ is fulfilled. 
A \keg\space is \textit{fulfilled} if all its  branches are fulfilled. 
A branch $\vartheta$ is \textit{complete} if either it is closed or it is open, fulfilled, and it does not contain any literal of type $x=y$, with $x$ and $y$ distinct variables. A \keg\space is \textit{complete} (resp., \emph{fulfilled}) if all its  branches are complete (resp., fulfilled or closed).  

 A $\coreflqsr$-interpretation $\M$ \emph{satisfies} a branch $\vartheta$ of a \keg\space (or, equivalently, $\vartheta$ \emph{is satisfied} by $\M$), and we write $\M \models \vartheta$, if $\M \models X$, for every formula $X$ occurring in $\vartheta$. 
 A $\coreflqsr$-interpretation $\M$ satisfies a  \keg\space $\mathcal{T}$ (or, equivalently, $\mathcal{T}$ \emph{is satisfied} by $\M$), and we write $\M \models \mathcal{T}$, if $\M$ satisfies a branch $\vartheta$ of $\mathcal{T}$. 
A branch $\vartheta$ of a \keg\space $\mathcal{T}$ is \emph{satisfiable} if there exists a $\coreflqsr$-interpretation $\M$ that satisfies $\vartheta$. A \keg\space is satisfiable if at least one of its branches is satisfiable. 
 

The procedure $\consistency$ takes care of literals of type $x=y$ occurring in the branches of $\T_\KB$ by constructing, for each open and fulfilled branch $\vartheta$ of $\T_\KB$ a substitution $\sigma_{\vartheta}$ such that $\vartheta\sigma_{\vartheta}$ does not contain literals of type $x=y$ with distinct $x,y$. Then, for every open and complete branch $\vartheta':=\vartheta\sigma_{\vartheta}$ of $\T_{\KB}$, the procedure  $\prochoplus$ constructs a decision tree $\DT_{\vartheta'}$ such that every maximal branch of $\DT_{\vartheta'}$ induces a substitution $\sigma'$ 
such that $\sigma_{\vartheta}\sigma'$ belongs to the answer set of $\psi_{Q}$ with respect to $\phi_{\KB}$. 

Specifically, the decision tree $\DT_{\vartheta'}$ is defined as follows.
Let $d$ be the number of literals in $\psi_Q$. Then $\DT_{\vartheta'}$ is a finite labelled tree of depth $d+1$ whose labelling satisfies the following conditions, for $i=0,\ldots,d$:
\begin{itemize}[itemsep=0.2cm]
\item[(i)]  every node of $\DT_{\vartheta'}$ at level $i$ is labelled with $(\sigma'_i, \psi_Q\sigma_{\vartheta}\sigma'_i)$; in particular, the root is labelled with
$(\sigma'_0, \psi_Q\sigma_{\vartheta}\sigma'_0)$, where $\sigma'_0$ is the empty substitution; 

\item[(ii)] if a node at level $i$ is labelled with $(\sigma'_i, \psi_Q\sigma_{\vartheta}\sigma'_i)$, then its $s$ successors, with $s >0$, are labelled with $\big(\sigma'_i\varrho^{q_{i+1}}_1, \psi_Q\sigma_{\vartheta}(\sigma'_i\varrho^{q_{i+1}}_1)\big),\ldots,\big(\sigma'_i\varrho^{q_{i+1}}_s, \psi_Q\sigma_\vartheta(\sigma'_i\varrho^{q_{i+1}}_s)\big)$, 
where $q_{i+1}$ is the $(i+1)$-st conjunct of $\psi_Q\sigma_{\vartheta}\sigma'_i$ and $\mathcal{S}_{q_{i+1}}=\{\varrho^{q_{i+1}}_1, \ldots, \varrho^{q_{i+1}}_s  \}$ is the collection of the substitutions\\ 
$\varrho = \{x_{v_1}/x_{o_1}, \ldots, x_{v_n}/ x_{o_n},  X^1_{c_1}/X^1_{C_1}, \ldots, X^1_{c_m}/ X^1_{C_m},$\\
\null\hfill $X^3_{r_1}/X^3_{R_1}, \ldots, X^3_{r_k}/X^3_{R_k},  X^3_{p_1}/X^3_{P_1}, \ldots, X^3_{p_h}/X^3_{P_h}\}$,\\ 
{\spaceskip  1.35em  \relax
with $\{x_{v_1}, \ldots, x_{v_n}\} = \varz(q_{i+1})$, $\{X^1_{c_1},\ldots, X^1_{c_m}\} = \varu(q_{i+1})$, and}\\
$\{X^3_{p_1},\ldots, X^3_{p_h},X^3_{r_1},\ldots, X^3_{r_k}\} = \vart(q_{i+1})$, such that $t=q_{i+1}\varrho$, for some  literal $t$ on $\vartheta'$. If $s = 0$, the node labelled with $(\sigma'_i, \psi_Q\sigma_{\vartheta}\sigma'_i)$ is a leaf node and, if $i = d$,  $\sigma_\vartheta\sigma'_i$ is added to $\Sigma'$.  In this case, the leaf node is contained in a non failing-branch and the substitution $\sigma_\vartheta\sigma'_i$ is a match for the query $\psi$.
\end{itemize}
The decision tree $\DT_{\vartheta'}$ is represented as a stack of its nodes. Initially the stack contains the root node $(\epsilon,\psi_Q\sigma_\vartheta)$ of  $\DT_{\vartheta'}$, as defined by condition (i). Then, iteratively, the following steps are executed. An element $(\sigma', \psi_Q\sigma_\vartheta\sigma')$ is popped out of the stack. If the last literal of the query $\psi_Q$ has not been reached, the successors of the current node are computed according to condition (ii) and inserted in the stack. Otherwise the current node must have the form $(\sigma',\lambda)$, with $\lambda$ the empty query, and the substitution $\sigma_\vartheta\sigma'$  is inserted in $\Sigma'$. Notice that, in case of a failing query match, the $\litqt$ computed at step 13 is empty. Since the \textsf{while}-loop 14--18 is not executed, no successor node is pushed in the stack. Thus, the failing branch is abandoned and the procedure selects another branch by means of a pop of one of its nodes from the stack.

The procedures $\consistency$ and $\prochoplus$ are shown next. 


{\scriptsize
\begin{algorithmic}[1]
\Procedure{$\consistency$}{$\phi_\KB$}
\State $\Phi_\KB :=\{\phi: \phi \mbox{ is a conjunct of } \phi_{\KB}\}$;
\State $\T_{\KB}$ := $\Phi_\KB$;
\State $\mathcal{E}:= \emptyset$;
\While{$\T_{\KB}$ is not fulfilled}
\parState{- select a not fulfilled open branch $\vartheta$ of $\T_{\KB}$ and a not fulfilled formula\\ \hspace*{.2cm}$\psi=(\forall{x_1})\ldots(\forall{x_m})(\beta_1 \vee \ldots \vee \beta_n)$ in $\vartheta$;}
%
%
%
%
%
\parState{ $\Sigma^{\KB}_\psi = \{\tau : \tau=\{x_1/x_{o_m}, \ldots, x_m/x_{o_m}\}\}$, where $\{x_1,\ldots,x_m\}=Q\var_0(\psi)$ and $\{x_{o_1},\ldots,x_{o_m}\}\in \varz(\phi_{\KB})$;}
\For{$\tau \in \Sigma^{\KB}_\psi$} 
\If{$\beta_i\tau \notin \vartheta$, for every $i=1,\ldots,n$}
\If{$\seqsnj$ is in $\vartheta$, for some $j \in \{1,\ldots,n\}$}
    \State - apply the $\egamma$ to $\psi$ and $\seqsnj$ on $\vartheta$; \label{procCon:Egamma}
\Else 
    \parState{- let $B^{\overline{\beta}\tau}$ be the collection of literals $\overline{\beta}_1\tau,\ldots,\overline{\beta}_n\tau$ present in $\vartheta$
and let\\ \hspace*{.15cm} $h$ be the lowest index such that $\overline{\beta}_h\tau \notin B^{\overline{\beta}\tau}$;}
    \parState{- apply the PB-rule to $\overline{\beta}_h\tau$ on $\vartheta$;} \label{procCon:pbrule}
\EndIf;
\EndIf;
\EndFor;
\EndWhile;
\For {$\vartheta$ in $\T_{\KB}$}
\If {$\vartheta$ is an open branch}
\label{procCon:whileBranch}
%

\State $\sigma_{\vartheta} := \epsilon$ (where $\epsilon$ is the empty substitution);

\State $\mathsf{Eq}_{\vartheta} := \{ \mbox{literals of type $x = y$, occurring in $\vartheta$}\}$;

\While{$\mathsf{Eq}_{\vartheta}$ contains $x = y$, with distinct $x$, $y$} \label{procCon:whileEQ}

    \State - select a literal $x = y$ in $\mathsf{Eq}_{\vartheta}$, with distinct $x$, $y$;
    
    \State $z :=$ $min_{\minord}(x,y)$ (with $\minord$ an arbitrary but fixed total order on $\mathsf{Var}_{0}(\phi_{\KB})$);
    
    \State $\sigma_{\vartheta} := \sigma_{\vartheta} \cdot \{x/z, y/z\}$;
    
    \State $\mathsf{Eq}_{\vartheta} := \mathsf{Eq}_{\vartheta}\sigma_{\vartheta}$;
\EndWhile;

\State $\mathcal{E} = {\mathcal{E} \cup \{(\vartheta,\sigma_\vartheta)\}}$;
\State $\vartheta := \vartheta\sigma_{\vartheta}$;
\EndIf;
\EndFor;
\State \Return $(\T_\KB, \mathcal{E})$;
\EndProcedure;
\end{algorithmic}
}
\vspace*{0.3cm}
{\scriptsize
\begin{algorithmic}[1]
	\Procedure{$\prochoplus$}{$\psi_Q$, $\mathcal{E}$}
	
	\State $\Sigma'$ := $\emptyset$;	
	\parWhile{$\mathcal{E}$ is not empty}
	\State - let  $(\vartheta,\sigma_\vartheta) \in \mathcal{E}$;
	\State - $\vartheta := \vartheta\sigma_{\vartheta}$;
\State - initialize $\mathcal{S}$ to the empty stack;

\State - push $(\epsilon, \psi_Q\sigma_\vartheta)$ in $\mathcal{S}$;

\While{$\mathcal{S}$ is not empty}
\State - pop $(\sigma', \psi_Q\sigma_\vartheta\sigma')$ from $\mathcal{S}$;

\If{$\psi_Q\sigma_\vartheta\sigma' \neq \lambda$}
\State - let $q$ be the leftmost conjunct of $\psi_Q\sigma_\vartheta\sigma'$;

\State $\psi_Q\sigma_\vartheta\sigma':= \psi_Q\sigma_\vartheta\sigma'$ deprived of $q$;
\State $\litqt := \{ t \in \vartheta : t=q\rho$, for some substitution $\rho \}$; \label{procHOP:compLit}

\While{$\litqt$ is not empty} \label{procHOP:whileLit}
\State - let $t \in \litqt$, $t=q\rho$;

\State $\litqt := \litqt \setminus \{t\}$;

\State - push $(\sigma'\rho, \psi_Q\sigma_\vartheta\sigma'\rho)$ in $\mathcal{S}$;
\EndWhile;
\Else
\State $\Sigma'$ := $\Sigma' \cup \{\sigma_\vartheta\sigma'\}$;
\EndIf;
\EndWhile;
   \State $\mathcal{E} := {\mathcal{E} \setminus \{(\vartheta,\sigma_\vartheta)\}}$; \label{procHOP:Edel}
	\EndparWhile;	
	\State \Return $\Sigma'$;
	\EndProcedure;
\end{algorithmic}
}
\normalsize

%
%
%
%

\subsection{Correctness of the procedures $\consistency$ and $\prochoplus$}
Correctness of the procedure $\consistency$ follows from Theorems \ref{teo:correctnessplus} and \ref{teo:completenessplus}, which show that $\phi_\KB$ is satisfiable if and only if  $\T_{\KB}$ is a non-closed \keg, whereas correctness of the procedure $\prochoplus$ is  proved by Theorem \ref{teo:proc2corrplusBis}, which shows that the output set  $\Sigma'$ is the HO-answer set of $\psi_Q$ w.r.t.\ $\phi_\KB$. 

Before stating (and proving) Theorems~\ref{teo:correctnessplus}, \ref{teo:completenessplus}, and \ref{teo:proc2corrplusBis}, we prove the following technical lemmas, which are needed for the proof of Theorem~\ref{teo:correctnessplus}.

\begin{lemma}\label{lemma:invariantplusBis}
\begin{sloppypar}
	Let $\vartheta$ be a branch of $\T_{\KB}$ selected at step $19$ of procedure $\consistency$ ($\phi_{\KB}$), let $\sigma_{\vartheta}$ be the associated substitution constructed during the execution of the \textsf{while}-loop 23--28, and let $\M = (D,M)$ be a $\coreflqsr$-interpretation satisfying $\vartheta$. Then
	\end{sloppypar} 
	\begin{equation}\label{eq:invariantplusBis}
	Mx = Mx\sigma_{\vartheta}, \mbox{ for every } x \in \mathsf{Var}_0(\vartheta),
	\end{equation}
	is an invariant of the \textsf{while}-loop 23--28. 
\end{lemma}

\begin{proof}
	We prove the thesis by induction on the number $i$ of iterations of the \textsf{while}-loop 23--28 of the procedure $\prochoplus$($\psi_{Q}$,$\phi_{\KB}$). For simplicity we indicate with $\sigma_{\vartheta}^{(i)}$ and with $Eq_{\sigma_{\vartheta}}^{(i)}$ the substitution $\sigma_{\vartheta}$ and the set $Eq_{\sigma_{\vartheta}}$calculated at iteration $i \geq 0$, respectively. 
	
	If $i = 0$, $\sigma_{\vartheta}^{(0)}$ is the empty substitution $\epsilon$ and thus (\ref{eq:invariantplusBis}) trivially holds. 
	
	Assume by inductive hypothesis that (\ref{eq:invariantplusBis}) holds at iteration $i \geq 0$. We want to prove that (\ref{eq:invariantplusBis}) holds at iteration $i+1$. 
	
	At iteration $i +1$, $\sigma_{\vartheta}^{(i+1)} = \sigma_{\vartheta}^{(i)} \cdot \{x/z,y/z\}$, where $z = \min_{\minord}\{x,y\}$ and $x = y$ is a literal in $Eq_{\sigma_{\vartheta}}^{(i)}$, with distinct $x,y$. We assume, without loss of generality, that $z$ is the variable $x$ (an analogous proof can be carried out assuming that $z$ is the variable $y$). By inductive hypothesis $Mw = Mw\sigma_{\vartheta}^{(i)}$, for every $w \in \mathsf{Var}_0(\vartheta)$. If $w\sigma_{\vartheta}^{(i)} \in \mathsf{Var}_0(\vartheta)\setminus \{y\}$, plainly $w\sigma_{\vartheta}^{(i)}$ and $w\sigma_{\vartheta}^{(i+1)}$ coincide and thus 
	$Mw\sigma_{\vartheta}^{(i)} = Mw\sigma_{\vartheta}^{(i+1)}$. Since $Mw = Mw\sigma_{\vartheta}^{(i)}$, it follows that $Mw = Mw\sigma_{\vartheta}^{(i+1)}$. 
	
	If $w\sigma_{\vartheta}^{(i)}$ coincides with $y$, we reason as follows. At iteration $i+1$, the variables $x,y$ are considered because the literal $x=y$ is selected from $Eq_{\sigma_{\vartheta}}^{(i)}$. 
	
	If $x=y$ is a literal belonging to $\vartheta$, then $Mx = My$. Given that $w\sigma_{\vartheta}^{(i)}$ coincides with $y$ and $w\sigma_{\vartheta}^{(i+1)}$ coincides with $x$,  $My = Mx$ implies  $Mw\sigma_{\vartheta}^{(i)} = Mw\sigma_{\vartheta}^{(i+1)}$. Since by inductive hypothesis $Mw = Mw\sigma_{\vartheta}^{(i)}$, it holds that $Mw = Mw\sigma_{\vartheta}^{(i+1)}$. 
	
	If $x  = y$ is not a literal occurring in $\vartheta$, then $\vartheta$ must contain a literal $x' = y'$ such that $x$ coincides with $x'\sigma_{\vartheta}^{(i)}$ and $y$ coincides with $y'\sigma_{\vartheta}^{(i)}$ at iteration $i$. Since $Mx' = My'$ and, by inductive hypothesis, $Mx' = Mx'\sigma_{\vartheta}^{(i)}$ and $My' = My'\sigma_{\vartheta}^{(i)}$, we have that $Mx = My$, and thus, by reasoning as above,  $Mw = Mw\sigma_{\vartheta}^{(i+1)}$. As (\ref{eq:invariantplusBis}) holds at each iteration of the \textsf{while}-loop, it follows that it is an invariant of the loop, as we wished to prove.\qed
\end{proof}

\begin{lemma}\label{lemma:correctnessplusBis}
	Let $\T_0,\ldots,\T_h$ be a sequence of \keg x such that $\T_0 = \phi_{\KB}$, and $\T_{i+1}$ is obtained from $\T_i$ by applying either the rule of step 11, or the rule of step 14, or the  substitution of step 30 of procedure $\consistency$($\phi_{\KB}$), for $i = 1,\ldots,h-1$. 
	If $\T_i$ is satisfied by a $\coreflqsr$-interpretation $\M$, then $\T_{i+1}$ is satisfied by $\M$ as well, for $i = 1,\ldots,h-1$.
\end{lemma}  
\begin{proof}
	Let $\M=(D,M)$ be a $\coreflqsr$-interpretation satisfying $\T_i$. Then $\M$ satisfies a branch $\bar{\vartheta}$ of $\T_i$. In case the branch $\bar{\vartheta}$ is different from the branch selected at step 5, if the $\egamma$ (step 11) or the PB-rule (14) is applied, or if a substitution for handling equalities (step 30) is applied, $\bar{\vartheta}$ belongs to $\T_{i+1}$ and therefore $\T_{i+1}$ is satisfied by $\M$. 
	In case $\bar{\vartheta}$ is the branch selected and modified to obtain $\T_{i+1}$, we have to consider the following two cases. 

\smallskip
\noindent\textbf{The branch $\bar{\vartheta}$ has been selected at step $6$ (and thus it is an open branch not yet fulfilled):}
		Let $\psi=(\forall{x_1})\ldots(\forall{x_m})(\beta_1 \vee \ldots \vee \beta_n)$ be the not fulfilled formula selected in $\bar{\vartheta}$, and $\tau=\{x_1/x_{o_1}, \ldots, x_m/x_{o_m}\}$ the substitution in $\Sigma^{\KB}_\psi$  chosen at step 8. If $\beta_i\tau \in \theta$, for some $i=i, \ldots, n$, then step $8$ proceeds with the next iteration. Otherwise, if step $10$ is executed, the $\egamma$ is applied to the formula $\psi=(\forall{x_1})\ldots(\forall{x_m})(\beta_1 \vee \ldots \vee \beta_n)$ and to the set of formulae $\seqsnj$ on the branch $\bar{\vartheta}$, generating the new branch $\bar{\vartheta'} := \bar{\vartheta} ; \beta_i\tau$. Since $\M \models \bar{\vartheta}$, we plainly have that $\M \models \psi=(\forall{x_1})\ldots(\forall{x_m})(\beta_1 \vee \ldots \vee \beta_n)$ and $\M \models \seqsnj$.
		
		Since $\M \models \bar{\beta}_j\tau$ , we have that $\M[x_1/Mx_{o_1}, \ldots, x_m/Mx_{o_m}]  \models  \bar{\beta}_j$, for $j \in \{1, \ldots, n\}\setminus\{i\}$.  
		
		Considering that $\M \models \psi=(\forall{x_1})\ldots(\forall{x_m})(\beta_1 \vee \ldots \vee \beta_n)$, then
\[\M[x_1/Mx_{o_1}, \ldots, x_m/Mx_{o_m}]  \models  \beta_1 \vee \ldots \vee \beta_n,
\]
so that $\M[x_1/Mx_{o_1}, \ldots, x_m/Mx_{o_m}]  \models  {\beta}_i$, namely, $\M \models \beta_i\tau$, as we wished to prove.  
		
		If step $14$ is performed, the PB-rule is applied on $\bar{\vartheta}$, originating the  branches  (belonging to $\T_{i+1}$) $\bar{\vartheta'} := \bar{\vartheta} ; \overline{\beta}_h$ and $\bar{\vartheta''} := \bar{\vartheta} ; \beta_h$. Since either $\M \models \beta_h$ or  $\M \models \overline{\beta}_h$, then either $\M \models \bar{\vartheta'}$ or $\M \models \bar{\vartheta''}$. Thus $\M$ satisfies $\T_{i+1}$, as we wished to prove. 

\smallskip

\noindent\textbf{The branch $\bar{\vartheta}$ has been selected at step $19$ (and thus it is an open and fulfilled branch not yet complete:} Once step $30$ is executed, the new branch $\bar{\vartheta} \sigma_{\bar{\vartheta}}$ is generated. Since $\M \models \bar{\vartheta}$ and, by Lemma \ref{lemma:invariantplusBis}, $Mx = Mx\sigma_{\bar{\vartheta}}$, for every $x \in \mathsf{Var}_0(\bar{\vartheta})$, then $\M \models \bar{\vartheta} \sigma_{\bar{\vartheta}}$, and therefore $\M$ satisfies $\T_{i+1}$, completing the proof of the lemma.\qed 
\end{proof}

\begin{theorem}\label{teo:correctnessplus}
	If $\phi_{\KB}$ is satisfiable, then $\T_{\KB}$ is not closed.
\end{theorem}
\begin{proof}
	Let us assume, for contradiction, that $\T_{\KB}$ is closed. Since $\phi_{\KB}$ is satisfiable, there exists a $\coreflqsr$-interpretation $\M$ satisfying every formula of $\phi_{\KB}$. Thanks to Lemma \ref{lemma:correctnessplusBis}, any \keg\space for $\phi_{\KB}$ obtained by applying either step 11, or step 14, or step 30 of the procedure $\consistency$ is satisfied by $\M$. Thus, $\T_{\KB}$ is satisfied by $\M$ as well. In particular, there exists a branch $\vartheta_c$ of $\T_{\KB}$ satisfied by $\M$. From our initial assumption that $\T_{\KB}$ is closed, it follows that the branch $\vartheta_c$ is closed as well and thus it must contain either both $A$ and $\neg A$, for some formula $A$, or a literal of type $\neg (x = x)$. But $\vartheta_c$ is satisfied by $\M$; hence, either $\M \models A$ and $\M \models \neg A$ or $\M \models \neg (x = x)$, which are clearly impossible. Thus, the \keg\space $\T_{\KB}$ must be not closed, proving the theorem.\qed  
\end{proof}

\begin{theorem}\label{teo:completenessplus}
	If $\T_{\KB}$ is not closed, then  $\phi_{\KB}$ is satisfiable.
\end{theorem}
\begin{proof}
	Since $\T_{\KB}$ is not closed, there must exist a branch $\vartheta'$ in $\T_{\KB}$ which is open and complete. 
	The branch $\vartheta'$ is obtained during the execution of the procedure $\consistency$ from an open fulfilled branch $\vartheta$ by applying to it the substitution $\sigma_{\vartheta}$ constructed during the execution of the \textsf{while}-loop at step 19 of the procedure. Thus, $\vartheta' = \vartheta\sigma_{\vartheta}$. Since each formula of $\phi_{\KB}$ occurs in $\vartheta$, to prove that $\phi_{\KB}$ is satisfiable, it is enough to show that $\vartheta$ is satisfiable. 
	
	Let us construct a $\coreflqsr$-interpretation $\M_{\vartheta}=(D_{\vartheta},M_{\vartheta})$ satisfying every formula $X$ occurring in $\vartheta$ and thus $\phi_{\KB}$. We put:
	\begin{itemize}
		\item $D_{\vartheta} \defAs \{x\sigma_{\vartheta} : x \in \mathsf{Var}_0(\vartheta) \}$; 
		\item $M_{\vartheta} x \defAs x\sigma_{\vartheta}$, \quad for every $x \in \mathsf{Var}_0(\vartheta)$; 
		\item $M_{\vartheta} X^1 \defAs \{x\sigma_{\vartheta} : x \in X^1 \mbox{ occurs in } \vartheta\}$, \quad for every $X^1 \in \mathsf{Var}_1(\vartheta)$;  
		\item $M_{\vartheta} X^3 \defAs \{\langle x\sigma_{\vartheta}, y\sigma_{\vartheta} \rangle : \langle x, y\rangle \in X^3  \mbox{ occurs in } \vartheta \}$, \quad for every $X^3 \in \mathsf{Var}_3(\vartheta)$. 
	\end{itemize}
	
	
	Next we show that $\M_{\vartheta}$ satisfies each formula in $\vartheta$. We shall proceed by structural induction and case distinction. 
	To begin with, we consider the case in which the literal $x = y$ occurs in $\vartheta$. By the very construction of $\sigma_{\vartheta}$, as described in  procedure $\consistency$, $x\sigma_{\vartheta}$ and $y\sigma_{\vartheta}$ have to coincide. Thus, $M_{\vartheta}x = x\sigma_{\vartheta} = y\sigma_{\vartheta} = M_{\vartheta} y$ and then $\M_{\vartheta} \models x=y$.
	
	Next, let us assume that the literal $\neg (z = w)$ occurs in $\vartheta$. If $z\sigma_{\vartheta}$ and $w\sigma_{\vartheta}$ coincide, namely they are the same variable, then the branch $\vartheta' = \vartheta\sigma_{\vartheta}$ must be closed, contradicting our initial hypothesis. Thus. $z\sigma_{\vartheta}$ and $w\sigma_{\vartheta}$ must be distinct variables and therefore $M_{\vartheta} z= z\sigma_{\vartheta} \neq w\sigma_{\vartheta} = M_{\vartheta}w$. It follows that $\M_{\vartheta} \not\models z = w$ and, therefore, $\M_{\vartheta} \models \neg(z = w)$, as we wished to prove. 
	
	If $x \in X^1$ occurs in $\vartheta$, then, by the very definition of $M_{\vartheta}$, we have $x\sigma_{\vartheta} \in M_{\vartheta}X^1$, namely $M_{\vartheta} x \in M_{\vartheta}X^1$. Thus, $\M_{\vartheta} \models x \in X^1$, as desired. 
	
	If $\neg(y \in X^1)$ occurs in $\vartheta$, then $y\sigma_{\vartheta}\notin M_{\vartheta}X^1$. Assume, by way of contradiction, that $y\sigma_{\vartheta}\in M_{\vartheta}X^1$. Then there is a literal $z \in X^1$ in $\vartheta$ such that $z\sigma_{\vartheta}$ and $y\sigma_{\vartheta}$ coincide. In this case the branch $\vartheta'$, obtained from $\vartheta$ by applying the substitution  $\sigma_{\vartheta}$ would be closed, contradicting our initial hypothesis. Thus, we have  $y\sigma_{\vartheta}\notin M_{\vartheta}X^1$, which implies $M_{\vartheta}y \notin M_{\vartheta}X^1$. Hence, $\M_{\vartheta} \not\models y \in X^1$, so that $\M_{\vartheta} \models \neg(y \in X^1)$. 
	
	If $\langle x,y\rangle \in X^3$ occurs in $\vartheta$, then, by the very definition of $M_{\vartheta}$, we have $\langle x\sigma_{\vartheta}, y\sigma_{\vartheta}\rangle \in M_{\vartheta}X^3$, that is, $\langle M_{\vartheta}x, M_{\vartheta}y\rangle \in M_{\vartheta}X^3$, so that $\M_{\vartheta} \models \langle x,y\rangle \in X^3$. 
	
	Next, assume that $\neg(\langle z,w\rangle \in X^3)$ occurs in $\vartheta$, but $\langle z\sigma_{\vartheta},w\sigma_{\vartheta}\rangle \in M_{\vartheta}X^3$. Then a literal $\langle z',w'\rangle \in X^3$ occurs in $\vartheta$ such that $z\sigma_{\vartheta}$ coincides with $z'\sigma_{\vartheta}$ and $w\sigma_{\vartheta}$ coincides with $w'\sigma_{\vartheta}$. But then, the branch $\vartheta' = \vartheta\sigma_{\vartheta}$ would be closed, a contradiction. Thus, we must have $\langle z\sigma_{\vartheta},w\sigma_{\vartheta}\rangle \notin M_{\vartheta}X^3$, that is $\langle M_{\vartheta}z,M_{\vartheta}w\rangle \notin M_{\vartheta}X^3$. Hence, $\M_{\vartheta} \not\models \langle x,y\rangle \in X^3$, yielding $\M_{\vartheta} \models \neg(\langle x,y\rangle \in X^3)$. 
	\begin{sloppypar}
	Finally, let $\psi\defAs (\forall{x_1})\ldots(\forall{x_m})(\beta_1 \vee \ldots \vee \beta_n)$ be a $\coreflqsr$-purely universal quantified formula of level 1 occurring in $\vartheta$. Since $\vartheta$ is fulfilled, then $\psi$ is fulfilled too, so that $\vartheta$ must contain the formula $\beta_i\tau$, for some $i=1, \ldots, n$ and for all $\tau$ in $\Sigma^{\KB}_\psi$.
%
	Let $\tau=\{x_1/x_{o_1}, \ldots, x_m/x_{o_m}\}$ be any substitution in $\Sigma^{\KB}_\psi$. By inductive hypothesis, we have $\M_\vartheta \models \beta_i\tau$, for some $i \in \{1, \ldots, n\}$. 
	Thus, $\M_\vartheta[x_1/Mx_{o_1}, \ldots, x_m/Mx_{o_m}]  \models  \beta_i$ and, \emph{a fortiori}, $\M_\vartheta [x_1/Mx_{o_1}, \ldots, x_m/Mx_{o_m}]  \models  \beta_1 \vee \ldots \vee \beta_n$. From the generality of $\tau$, it follows that $\M_\vartheta  \models  (\forall x_1) \ldots (\forall x_m) (\beta_1 \vee \ldots \vee \beta_n)$, namely, $\M_{\vartheta}\models\psi$.
	\end{sloppypar}
	In conclusion, we have shown that $\M_{\vartheta}$ satisfies each formula in $\vartheta$ and, in particular, all the formulae in $\phi_{\KB}$, as we wished to prove. \qed
\end{proof}

\begin{lemma}\label{lemma:proc2plusBis}
	Let $\psi_Q \defAs q_1 \wedge \ldots \wedge q_d$ be a HO $\coreflqsr$-conjunctive query, $\Sigma'$ the output of $\prochoplus$($\psi_Q$, $\mathcal{E}$), and $\vartheta'$  an open and complete branch of $\T_\KB$. Then, for any substitution $\sigma'$, we have: 
\[
\sigma' \in \Sigma' \iff \{ q_1 \sigma', \ldots, q_d\sigma' \} \subseteq \vartheta'\,.
\]
\end{lemma}
\begin{proof} For the necessity part, assume that $\sigma' \in \Sigma'$. Then $\sigma'=\sigma_\vartheta\sigma'_d$ and the decision tree $\DT_{\vartheta'}$ contains a branch $\eta$ of length $d+1$ having $(\sigma'_d, \lambda)$ as leaf. Specifically, the branch $\eta$ consists of the following nodes:\\
\centerline{	
	$(\epsilon, q_1\sigma_\vartheta \wedge \ldots \wedge q_d\sigma_\vartheta )$, $( \rho^{(1)}, q_2\sigma_\vartheta\rho^{(1)} \wedge \ldots \wedge q_d\sigma_\vartheta\rho^{(1)} )$, $\ldots$, $( \rho^{(1)} \cdots \rho^{(d)}, \lambda)$,}
so that $\sigma'= \sigma_\vartheta \rho^{(1)} \cdots \rho^{(d)}$. Consider the node \\
\centerline{	
	$ (\rho^{(1)} \cdots \rho^{(i+1)}, q_{i+2}\sigma_\vartheta \rho^{(1)} \cdots \rho^{(i+1)} \wedge \ldots \wedge q_d\sigma_\vartheta \rho^{(1)} \cdots \rho^{(i+1)})$ 
}
obtained from the father node \\
\centerline{	
	$(\rho^{(1)} \cdots \rho^{(i)}, q_{i+1}\sigma_\vartheta \rho^{(1)} \cdots \rho^{(i)} \wedge \ldots \wedge q_d\sigma_\vartheta \rho^{(1)} \cdots \rho^{(i)}),$ 
}	
as $q_{i+1}\sigma_\vartheta \rho^{(1)} \cdots \rho^{(i)} = t$, for some $t \in \vartheta'$. The literal $q_{i+1}\sigma_\vartheta \rho^{(1)} \cdots \rho^{(i)}$ is ground,  therefore it coincides with $q_{i+1}\sigma'$. Thus, $q_{i+1} \sigma'=t$, and hence $q_{i+1}\sigma' \in \vartheta'$. By induction on $i=0, \ldots, d-1$, it therefore follows $\{ q_1\sigma', \ldots, q_d \sigma' \} \subseteq \vartheta'$, as we wished to prove. 
	

For the sufficiency part, we have to show that the decision tree $\DT_{\vartheta'}$ constructed by procedure $\prochoplus$($\psi_{Q}$, $\mathcal{E}$) has a branch $\eta$ of length $d+1$ having as leaf a node $(\sigma'_d,\lambda)$ such that $\sigma'=\sigma_\vartheta\sigma'_d$, where  $\sigma_{\vartheta}$ is the substitution such that $\vartheta'=\vartheta\sigma_{\vartheta}$, computed by procedure $\consistency$. Let $(\epsilon, q_1\sigma_\vartheta \wedge \ldots \wedge q_d\sigma_\vartheta )$ be the root of the decision tree $\DT_{\vartheta'}$. At step $33$ of procedure $\consistency$, the node $(\epsilon, q_1\sigma_\vartheta \wedge \ldots \wedge q_d\sigma_\vartheta )$ is popped out from the stack and the conjunct $q=q_1\sigma_\vartheta$ is selected. Then, all the elements of the set $Lit^{\vartheta'}_q$ are considered, namely all the literals $t$  in $\vartheta'$ such that $t=q_1\sigma_{\vartheta}\rho$, for some substitution $\rho$. Among them, we have also the literal $q_1\sigma'=q_1\sigma_\vartheta\sigma'_d$. Let us put $\rho^{(1)}=\sigma'_d$. At step 17, the procedure $\consistency$ pushes the node $(\sigma'_d, q_2\sigma_\vartheta\sigma'_d \wedge \ldots \wedge q_d\sigma_\vartheta\sigma'_d)$ in the stack. Then, also the conjuncts $q_2\sigma_\vartheta\sigma'_d, \ldots, q_d\sigma_\vartheta\sigma'_d$ are processed sequentially. Since each of them coincides with a literal on $\vartheta'$, we have $\rho^{(2)}= \ldots =\rho^{(d)} = \epsilon$.  Considering that $\sigma'_d\epsilon = \sigma'_d$, it follows that the procedure $\consistency$ builds 
	the following sequence of nodes 
\begin{gather*}
(\epsilon, q_1\sigma_\vartheta \wedge \ldots \wedge q_d\sigma_\vartheta )\\
(\sigma'_d, q_2\sigma_\vartheta\sigma'_d \wedge \ldots \wedge q_d\sigma_\vartheta\sigma'_d)\\
\vdots\\
(\sigma'_d, q_{d}\sigma_\vartheta\sigma'_d)\\
(\sigma'_d, \lambda)
\end{gather*}
%
forming a branch $\eta$ of length $d+1$ of $\DT_{\vartheta'}$. Since $\eta$ has as the node $(\sigma'_d, \lambda)$ as leaf, we have $\sigma_\vartheta\sigma'_d=\sigma' \in \Sigma'$, as we wished to prove. \qed

\end{proof}

\begin{theorem}\label{teo:proc2corrplusBis}
	Let $\Sigma'$ be the set of substitutions returned by the call to procedure $\prochoplus(\psi_Q$, $\mathcal{E}$). Then $\Sigma'$ is the HO-answer set of $\psi_Q$ w.r.t.\ $\phi_\KB$.
\end{theorem}

\begin{proof}
It is enough to show that the following two assertions hold:
	\begin{enumerate}[label=(\alph*)]
		\item\label{one} if $\sigma' \in \Sigma'$, then $\sigma'$ is an element of the HO-answer set of $\psi_Q$ w.r.t.\ $\phi_\KB$;
		\item\label{two} if $\sigma'$ is a substitution of the HO-answer set of $\psi_Q$ w.r.t.\ $\phi_\KB$, then $\sigma' \in \Sigma'$.
	\end{enumerate}
	We first prove assertion \ref{one}. Let $\sigma' \in \Sigma'$, and let $\vartheta'=\vartheta\sigma_\vartheta$ be an open and complete branch of $\T_\KB$ such that $\DT_{\vartheta'}$ contains a branch $\eta$ of $d+1$ nodes whose leaf is labelled $(\sigma'_d, \lambda)$, where $\sigma' = \sigma_{\vartheta}\sigma'_d$. By Lemma \ref{lemma:proc2plusBis}, we have $\{  q_1\sigma', \ldots, q_d\sigma' \} \subseteq \vartheta'$.

	Then $\M_\vartheta \models q_{i}\sigma'$, for $i = 1,\ldots,d$, where $\M_\vartheta$ is $\coreflqsr$-interpretation associated with $\vartheta$, satisfying every formula $X$ occurring in $\vartheta$, and constructed as shown in Theorem~\ref{teo:completenessplus}. Hence, $\M_\vartheta \models \psi_Q\sigma'$, and since $\M_\vartheta \models \phi_\KB$,  we plainly have $\M_\vartheta \models \phi_\KB \wedge \psi_Q\sigma'$.  Thus, $\sigma'$ is a substitution of the HO-answer set of $\psi_Q$ w.r.t.\ $\phi_\KB$, proving \ref{one}.

Next we prove that also assertion \ref{two} holds. Let $\sigma'$ be a substitution belonging to the HO-answer set of $\psi_Q$ w.r.t.\ $\phi_\KB$. Hence, there exists a $\coreflqsr$-interpretation $\M $ such that $\M \models \phi_\KB \wedge \psi_Q\sigma'$.
	Assume for contradiction that $\sigma' \notin \Sigma'$. Then, by Lemma \ref{lemma:proc2plusBis},  $\{q_1\sigma', \ldots, q_d\sigma' \} \not\subseteq\vartheta'$, for every open and complete branch $\vartheta'$ of $\T_\KB$. In particular, for any given open and complete branch $\vartheta'$ of $\T_\KB$, there exists an index $i \in \{1,\ldots,d\}$ such that $q_i\sigma' \notin \vartheta'$, i.e.,  $q_i\sigma'\notin \vartheta\sigma_{\vartheta}$, and thus $\M_\vartheta \not\models q_i\sigma'$, with $\M_\vartheta$ an $\coreflqsr$-interpretation associated to $\vartheta$, defined as illustrated in Theorem \ref{teo:completenessplus}.
Therefore, by the generality of $\vartheta'=\vartheta\sigma_\vartheta$, it follows that every $\M_\vartheta$ satisfying  $\T_\KB$  (as shown in  Theorem \ref{teo:completenessplus}), and thus $\phi_\KB$, does not satisfy $\psi_Q\sigma'$. Since  we can prove the satisfiability of  $ \phi_\KB \wedge \psi_Q\sigma'$ by restricting our interest to the interpretations $\M_\vartheta$ associated to the branches  $\vartheta$ of the tableau $\T_\KB$ and defined as in the proof of Theorem \ref{teo:completenessplus}, it turns out that $\sigma'$ is not a substitution belonging to the HO-answer set of $\psi_Q$ w.r.t.\ $\phi_\KB$, which is a contradiction. Thus, assertion \ref{two} must hold.
	
	Having proved assertions \ref{one} and \ref{two}, we can conclude that $\Sigma'$ and the answer set of $\psi_Q$ w.r.t.\ $\phi_\KB$ coincide, proving the theorem.\qed
\end{proof}

\subsection{Termination of the procedures $\consistency$ and $\prochoplus$}
Termination of the procedure $\consistency$ is based on the fact that the \textsf{while}-loops 5--18 and 19--32 terminate. In addition, the procedure $\prochoplus$ terminates, provided that the \textsf{while}-loop 8--22 terminates.

Concerning termination of the \textsf{while}-loop 5--18 of $\consistency$, our proof is grounded on the following facts. The loop selects iteratively a not fulfilled  branch $\vartheta$ and a $\coreflqsr$-purely universal quantified formula of level 1  $\psi=(\forall{x_1})\ldots(\forall{x_m})(\beta_1 \vee \ldots \vee \beta_n)$ occurring in it. Since the sets $Q\var_0(\psi)$ and $\varz({\phi_{\KB}})$ are finite, line 7 builds a finite set $\Sigma^{\KB}_{\psi}$ containing finite substitutions $\tau$. The internal \textsf{for}-loop 8--17 selects iteratively an element $\tau$ in $\Sigma^{\KB}_{\psi}$. The $\egamma$ and PB-rule are applied only if $\beta_i\tau \notin \vartheta$, for all $i=1,\ldots,n$. In particular, if the $\egamma$ is applied on $\vartheta$, the procedure adds $\beta_i\tau$ in $\vartheta$, for some $i=1,\ldots,n$.
 In case the PB-rule is applied on $\vartheta$, two branches are generated. On one branch the procedure adds $\beta_i\tau$, for some $i=1,\ldots,n$, whereas on the other one it adds  $\bar{\beta}_i\tau$, so that the set $B^{\overline{\beta}\tau}$ gains $\overline{\beta}_i\tau$ as a new element. After at most $n-1$ applications of the PB-rule, $\big|B^{\overline{\beta}\tau}\big|$ gets equal to $n-1$ and the $\egamma$ is applied. Since the set $\Sigma^{\KB}_{\psi}$ is finite, the \textsf{for}-loop 8--17 terminates after a finite number of steps. After the last iteration of the \textsf{for}-loop, $\vartheta$ contains $\beta_i\tau$, for some $i=1,\ldots,n$ and for all $\tau$, thus $\psi$ gets fulfilled. Since $\phi_{\KB}$ contains a finite number of formulae $\psi$, the \textsf{while}-loop 5--18 terminates in a finite number of steps, as we wished to prove.

Termination proofs for the \textsf{while}-loop 19--32 of  $\consistency$ and of the \textsf{while}-loop 8--22 of  $\prochoplus$ are analogous to the one of the \textsf{while}-loop 14--44 of $\procho$ in \cite{RR2017}.

\subsection{Complexity issues}
Next, we provide some complexity results.

Concerning the procedure $\consistency$, we reason as follows. Let $\psi$ be any $\coreflqsr$-purely universal quantified formula of level 1 in $\Phi_{\KB}$ (see line 2 of the procedure $\consistency$ for the definition of $\Phi_{\KB}$). Let $r$ be the maximum number of universal quantifiers in $\psi$, $\ell$ the maximum number of literals in $\psi$, and $k \defAs |\varz(\Phi_{\KB})|$. It easily follows that $|\Sigma^{\KB}_{\psi}|=k^r $. Since the maximum number of literals contained in $\psi$ is $\ell$, the procedure applies $\ell-1$ times the PB-Rule and one time the $\egamma$ to $\psi\tau_j$, for $j=1,\ldots,k^r$. Thus $\psi$ generates a \keg\space of height $\mathcal{O}(\ell k^r)$.
Assuming that $m$ is the number of $\coreflqsr$-purely universal quantified formulae of level 1 in $\Phi_{\KB}$, the maximum height of the \keg\space (which corresponds to the maximum size of the models of $\phi_{\KB}$ that are constructed as illustrated in Theorem \ref{teo:completenessplus}) is $\mathcal{O}( m \ell k^r)$ and the maximum number of leaves of the \keg, i.e., the maximum number of such models of $\phi_{\KB}$ is $\mathcal{O}(2^{m \ell    k^r})$.   Notice that the construction of $\mathsf{Eq}_{\vartheta}$ and of $\sigma_{\vartheta}$ in the lines 19--32 of procedure $\consistency$ takes $\mathcal{O}( m \ell k^r)$-time, for each branch $\vartheta$.

Let $\eta(\T_{\KB})$ and $\lambda(\T_{\KB})$ be, respectively, the height of $\T_{\KB}$ and the number of leaves of $\T_{\KB}$ computed by the procedure $\consistency$. Plainly, $\eta(\T_{\KB}) = \mathcal{O}(\ell  m  k^r)$  and $\lambda(\T_{\KB})= \mathcal{O}(2^{\ell  m  k^r})$, as computed above.

It is easy to verify that $s=\mathcal{O}(\ell m  k^r)$ is the maximum branching of $\DT_\vartheta$. 
 Since the height of $\DT_\vartheta$ is $h$, where $h$ is the number of literals in $\psi_Q$, and the successors of a node are computed in $\mathcal{O}(\ell m k^r)$ time, it follows that the number of leaves in $\DT_\vartheta$ is $\mathcal{O}(s^{h})=\mathcal{O}((\ell m  k^r)^{h})$ and that they can be computed  in $\mathcal{O}( s^{h} \cdot  \ell m k^r \cdot h) = \mathcal{O}(h \cdot (\ell m k^r)^{(h+1)})$-time. Finally, since we have $\lambda(\T_{\KB})$ of such decision trees, the answer set of $\psi_{Q}$ w.r.t.\\ $\phi_\KB$ can be computed in time
$\mathcal{O}(h \cdot (\ell m k^r)^{(h+1)} \cdot\lambda(\T_{\KB})) =\mathcal{O}( h \cdot (\ell m  k^r)^{(h+1)} \cdot 2^{\ell  m  k^r})$.

In consideration of the fact that the sizes of $\phi_\KB$ and $\psi_{Q}$ are polynomially related to those of $\KB$ and of $Q$, respectively (see the proof of Theorem 1 in \cite{RR2017ext} for details on the reduction), the HO-answer set of $Q$ with respect to $\KB$ can be computed in double-exponential time. If $\KB$ contains neither role chain axioms nor qualified cardinality restrictions, the maximum number of universal quantifiers in $\phi_{\KB}$, namely $r$, is a constant (in particular $r = 3$), and thus our \textit{HOCQA} problem can be solved in EXPTIME. 
Such upper bound compares favourably to the complexity of the usual CQA problem for a wide collection of DLs such as the Horn fragment of $\mathcal{SHOIQ}$ and of $\mathcal{SROIQ}$ which are, respectively, EXPTIME- and 2EXPTIME-complete in combined complexity (see \cite{Ortiz:2011:QAH:2283516.2283571} for details). 

\section{Remarks on different versions of the algorithm}

The C++ implementation of the algorithm presented in this paper, called \kegs-system, is more efficient than the prototype (KE-system) introduced in \cite{cilc17}. The main motivation behind such a performance improvement relies on the introduction of the $\egamma$ (see Fig.~\ref{exprule}) that acts on the $\coreflqsr$-purely universal quantified formulae in the KB by systematically instantiating them and applying the standard E-rule (elimination rule) on-the-fly. The $\egamma$ replaces the preliminary phase of systematic expansion of the $\coreflqsr$-purely universal quantified formulae in the KB and the subsequent application of the E-rule implemented by the KE-system presented in \cite{RR2017}. The \kegs-system turns out be more efficient also than an implementation (FO KE-system) of the FO \ke\space in \cite{dagostino94} that applies the standard $\gamma$- and E-rules. Incidentally, it  turns out that the  KE-system and the FO KE-system have similar performances.


All the three systems take as input  an OWL ontology also admitting SWRL rules and serialized in the OWL/XML syntax, satisfying the requirements of a $\shdlssx$-KB. Such ontologies are parsed in order to produce the internal coding of all axioms and assertions of the ontology in set-theoretic terms as a list of strings by exploiting the function $\theta$ used in \cite{RR2017} to map $\shdlssx$-KBs to $\coreflqsr$-formulae. Each  string represents either a $\coreflqsr$-quantifier-free literal or a $\coreflqsr$ purely universal quantified formula in CNF, whose quantifiers have been moved as inward as possible and renamed in such a way as to be pairwise distinct. Data-structures exploited by the three systems are implemented in a similar way. The interested reader is referred to \cite{cilc17} for details. We point out that $\coreflqsr$-quantified variables and $\coreflqsr$-free variables are collected into two separate vectors and stored in order of appearance in the KB. These vectors ensure that the individuals used for the expansion of the universally quantified formulae are selected in the same order for all the three systems. This fact  guarantees that the number of branches of the three systems is the same, a key-aspect in the evaluation of their performances. In fact, in a \ke-based system, the number of branches coincides with the number of distinct models that each system computes in order to saturate the KB. Since the number of distinct branches is the same and the PB-rule is the same in all the three systems, the difference of performance among them is only due to the expansion rules.

\begin{example}
Let  
\begin{multline*}
\phi_{\KB} \defAs \neg ( \langle x_{\text{\it Italy}}, x_{\text{\it Rome}} \rangle \in X^3_{\text{\it locatedIn}} ) \wedge (\forall z_1) ( \langle z_1, z_1 \rangle \in X^3_{\text{\it isPartOf}} ) \\ {} \wedge (\forall z_1)(\forall z_2)(  \neg ( \langle z_1, z_2 \rangle \in X^3_{\text{\it locatedIn}}) \vee \langle z_1, z_2 \rangle \in X^3_{\text{\it isPartOf}})
\end{multline*}
be a $\shdlssx$-KB and let $\psi_Q= \langle x_{Rome}, x_{Italy} \rangle \in X^3_{r}$ be a HO-$\shdlssx$ conjunctive query. Fig.~\ref{kegfig} shows a \ke\space and a \keg\space for the answer set of $\psi_{Q}$ w.r.t.\ $\phi_{KB}$. Since the FO \ke\space can be represented along the same lines as the \ke\space, we refrain from reporting it. 
\end{example}



\begin{figure}
    \begin{framed}
  	\centering \includegraphics[scale=0.48]{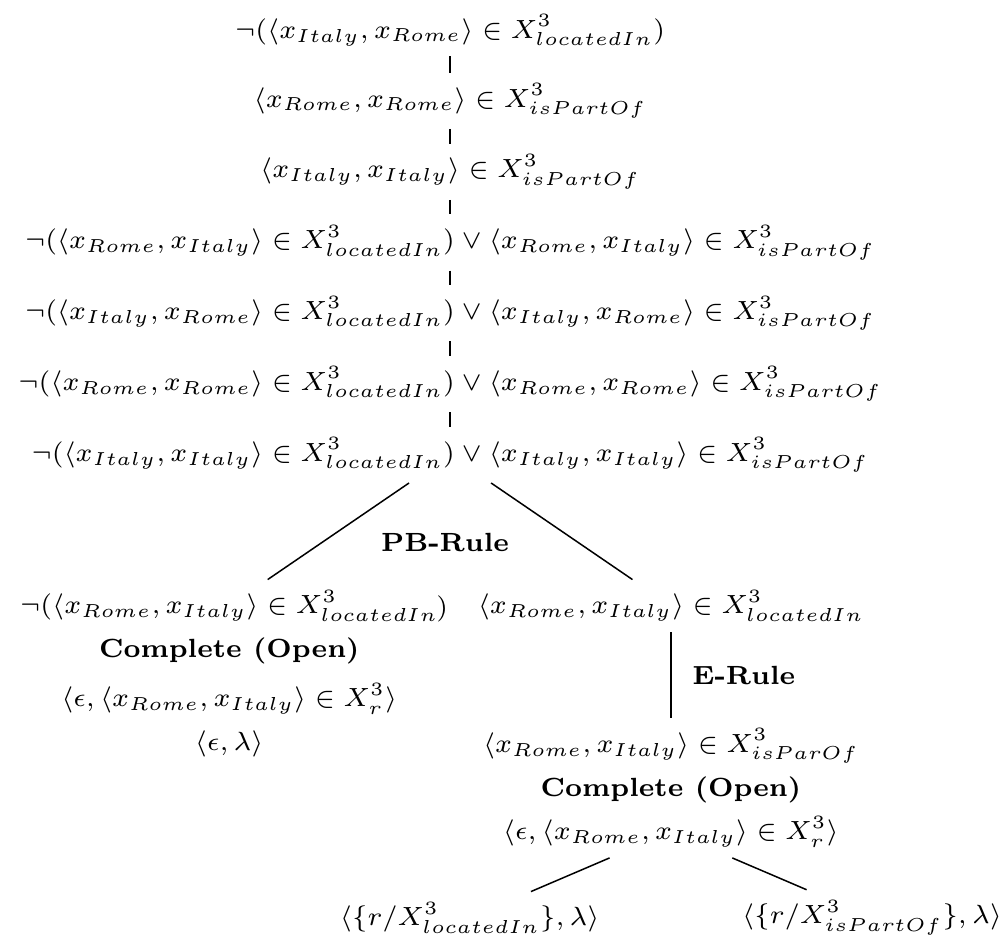} 
  	\end{framed}\vspace{-0.65cm}
  \begin{framed}
  	\centering \includegraphics[scale=0.48]{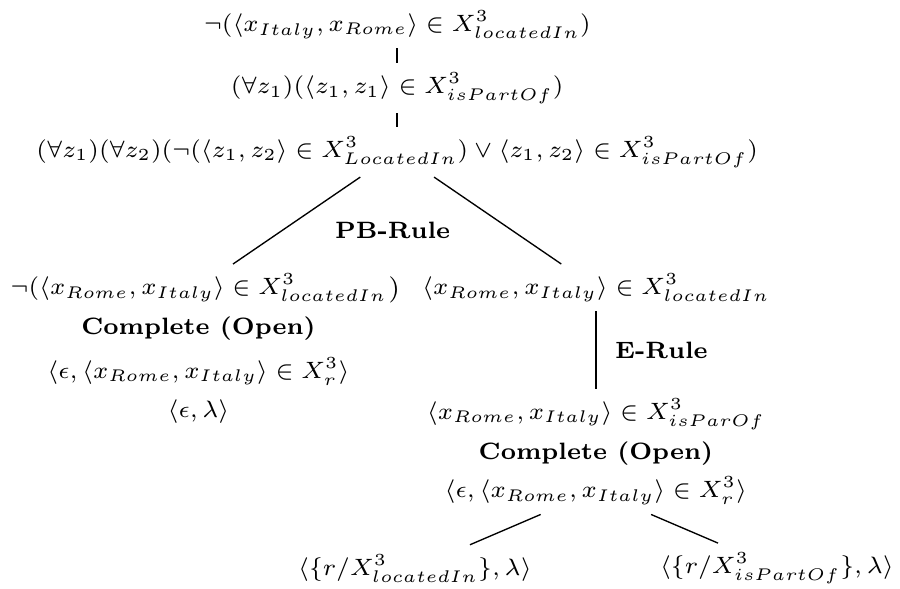}  
  \end{framed} 
  \vspace{-0.5cm}
  	\caption{\ke\space and \keg\space for the answer set of $\psi_Q$ w.r.t.\ $\phi_{\KB}$.}
  	\label{kegfig}
  \end{figure}

The metric used in the benchmarking is the number of models of the input KB computed by the reasoners and the time required to compute such models.

The $\shdlssx$-KBs considered in the tests have the following simple form:
\small
\begin{multline*}
\Phi_{\KB} \defAs \Big\{x_{a} \in X^1_D,~ x_{b} \in X^1_D,~ x_{c} \in X^1_D,~ x_{d} \in X^1_D, \\(\forall z)(\forall z_1)\big( (z \in X^1_{A} \wedge \langle z, z_1 \rangle  \in X^3_{P}\wedge z_1 \in X^1_{B} \wedge \langle z, z_1 \rangle \in X^3_{P_1}  ) \rightarrow z_1 \in X^1_{C}\big) \Big\}.
\end{multline*}
\normalsize
The KB $\Phi_{\KB}$ generates more than $10^6$ open branches which are computed in about $2$ seconds using the \kegs-system and in about $6$ seconds using the other systems. As shown in Fig.~\ref{bench}, the \kegs-system has a better performance than the other two up to about $400\%$, even if in some cases (lowest part of the plot) the performances of the three systems are comparable.  Thus the \kegs-system is always convenient, also because the collection of expansions of $\shdlssx$-purely universal quantified formulae of level 1  (exponential in the size of the KB) is not stored in memory.

\vspace*{-0.5cm}
\begin{figure}[H] 
	\centering \includegraphics[scale=0.6]{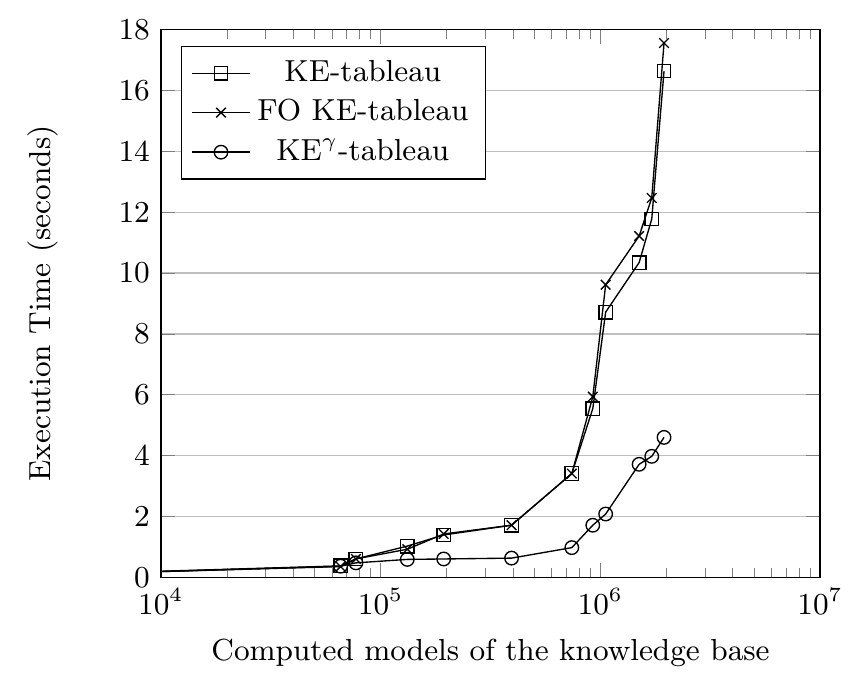}
	\caption{Comparison among the KE-system, FO KE-system, and \kegs-system.}	\label{bench}
\end{figure}

The benchmarking process is based on a huge amount of KBs of different size and kind, constructed ad-hoc for the purpose of comparing the three mentioned  systems, and on some real-world ontologies developed by the authors.

%
%
%
%

\section{Conclusions and future work}

We presented an improvement, called \keg, of the \ke\ in \cite{RR2017} for the most widespread reasoning tasks for $\shdlssx$-TBoxes and $\shdlssx$--ABoxes. 
These reasoning problems are addressed by translating $\shdlssx$-KBs and queries in terms of formulae of the set-theoretic language $\coreflqsr$. The procedure introduced in this paper generalizes the KE-elimination rule in such way as to incorporate the $\gamma$-rule, that is the expansion rule handling  universally quantified formulae. The \keg\space procedure has remarkable aftermath, since its implementation is markedly more efficient in terms of space and execution time than the KE-system \cite{cilc17} and the implementation (FO KE-system) of the FO \ke \cite{dagostino94}, as observed in our experimental tests.  

We plan to modify the set-theoretic fragment underpinning the reasoner so as to include a restricted version of the operator of relational composition in order to be able to reason with DLs admitting full existential and universal quantification. Results and notions presented in \cite{DBLP:conf/cade/CristiaR17} will be of inspiration for such a task. We also intend to improve our reasoner so as to deal with the reasoning problem of ontology classification. We shall compare the resulting reasoner with existing well-known reasoners such as Hermit \cite{ghmsw14HermiT} and Pellet \cite{PelletSirinPGKK07}, providing also some benchmarking. In addition, we plan to allow data type reasoning by either integrating existing solvers for the Satisfiability Modulo Theories (SMT) problem or by designing \emph{ad hoc} new solvers. 
Finally, as each branch of a \keg\space can be independently computed by a single processing unit, we plan to implement a  parallel version of the software by using the Nvidia CUDA framework.

\bibliographystyle{plain}
\bibliography{biblioext} 
\end{document}